\documentclass[journal,onecolumn,11pt]{IEEEtran}

\usepackage{amsmath}
\usepackage{amssymb}
\usepackage{mathtools} 
\usepackage{soul}
\usepackage{enumitem}
\usepackage[normalem]{ulem}
\usepackage{bm,dsfont}
\usepackage{xfrac}
\usepackage[hidelinks]{hyperref} 
\usepackage{multicol}
\usepackage{cancel}
\usepackage{makecell}
\usepackage{colortbl} 
\usepackage{float}
\usepackage{subcaption}
\usepackage{tikz}
\usetikzlibrary{arrows,automata}
\usepackage{color,soul}
\usepackage{graphicx}
\usepackage{comment}
\usepackage{colordvi}
\usepackage{bm} 

\usepackage{amsthm}
\newtheorem{theorem}{Theorem}
\newtheorem{definition}{Definition}
\newtheorem{proposition}{Proposition}
\newtheorem{lemma}{Lemma}
\newtheorem{corollary}{Corollary}

\newtheorem{innercustomgeneric}{\customgenericname}
\providecommand{\customgenericname}{}
\newcommand{\newcustomtheorem}[2]{%
  \newenvironment{#1}[1]
  {%
   \renewcommand\customgenericname{#2}%
   \renewcommand\theinnercustomgeneric{##1}%
   \innercustomgeneric
  }
  {\endinnercustomgeneric}
}

\newcustomtheorem{customthm}{Theorem}
\newcustomtheorem{customlemma}{Lemma}
\newcustomtheorem{customprop}{Proposition}

\newtheoremstyle{abcd}
  {}
  {}
  {\itshape}
  {}
  {\bfseries}
  {.}
  {.5em}
  {}
\theoremstyle{abcd}

\newcommand{\ie}{\emph{i.e., }}

\begin{document}
%
\title{Constrained Correlated Equilibria}
%
%
%

\author{Omar~Boufous, Rachid~El-Azouzi, Mikaël Touati, Eitan~Altman,~and Mustapha~Bouhtou
\thanks{O. Boufous is with Orange Labs, Châtillon, France and University of Avignon, Avignon, France. E-mail:{\tt~omar.boufous@alumni.univ-avignon.fr}.}
\thanks{R. El-Azouzi is with CERI/LIA, University of Avignon, Avignon, France. E-mail: {\tt~rachid.elazouzi@univ-avignon.fr}.}
\thanks{E. Altman is with the INRIA, Sophia Antipolis, France. E-mail: {\tt~eitan.altman@inria.fr}.}
\thanks{M. Touati and M. Bouhtou are with Orange Labs, Châtillon, France. E-mails:{\tt~ \{mikael.touati, mustapha.bouhtou\}@orange.com}.}
}

\maketitle

\begin{abstract}
\textcolor{black}{This paper introduces constrained correlated equilibrium, a solution concept combining correlation and \textcolor{black}{coupled} constraints in finite non-cooperative games.}
\textcolor{black}{In the general case of an arbitrary correlation device and coupled constraints in the extended game, we study the conditions for equilibrium. }
 In the particular case of constraints induced by a feasible set of probability distributions over action profiles, \textcolor{black}{we first} show that canonical correlation devices are sufficient to characterize the set of constrained correlated equilibrium distributions and provide conditions of their existence.
\textcolor{black}{Second, it is shown that} constrained correlated equilibria of the mixed extension of the game do not \textcolor{black}{lead to} additional equilibrium distributions. 
\textcolor{black}{Third,} we show that the constrained correlated equilibrium distributions may not belong to the polytope of correlated equilibrium distributions. \textcolor{black}{Finally, we illustrate these results through numerical examples.}
\end{abstract}


%
\IEEEpeerreviewmaketitle

%
%
%
%

\setcounter{footnote}{0} 

\section{Introduction}\label{sec:introduction}




\setcounter{footnote}{0}  

Correlated equilibria \cite{AUMANN1974} have been introduced as a generalization of Nash equilibria with appealing game-theoretic and Bayesian foundations \cite{aumann1987correlated}.
A correlated equilibrium of a non-cooperative game is \textcolor{black}{a pure} Nash equilibrium of an extension, called extended game,  by a randomizing structure (also known as a correlation device \cite{forges2020correlated}).
In the extended game, a Nash equilibrium induces a probability distribution over action profiles, called correlated equilibrium distribution.
One of the most important results on correlated equilibria shows that canonical devices (those randomizing over action profiles) are necessary and sufficient to characterize the polytope of correlated equilibrium distributions \cite{aumann1987correlated}.
Another major result of this influential paper shows that correlated equilibrium distributions emerge from the Bayesian rationality of the players, thus providing Bayesian decision-theoretic foundations to the concept.
 However, if for theoretical or practical reasons, some strategies (in the original game or its extended version) are forbidden or unfeasible, the standard definition of correlated equilibrium is no longer adequate as it may characterize equilibrium strategies that include unfeasible profiles or even exclude certain profiles that one would intuitively be considered as equilibria.
In fact,  there are extended games with coupled constraints such that  
 a strategy profile is not a correlated equilibrium (\ie not a pure Nash equilibrium of the extended game), but for every player, any unilateral deviation from this profile is either not profitable or leads to an unfeasible profile.
This naturally raises the question of whether or not it can be considered as a constrained correlated equilibrium and how to define a relevant solution concept for correlation with constraints in non-cooperative games.

\textcolor{black}{
No game-theoretic solution has yet been proposed to answer these questions despite the long-standing interest in generalized Nash equilibria \cite{facchinei2007generalized}, correlated equilibria and the potential of coordination in engineering systems or applications with coupled constraints (as an example, \cite{kulkarni2017games} considers coordination among the decision-takers as a key development of energy systems). One of the motivations for studying correlated equilibria with coupled constraints is the importance of correlated equilibrium in incorporating some form of coordination between players. Indeed, correlated equilibria extend the set of Nash equilibria and brings Pareto improvements over Nash equilibria, but the existence of other equilibria also exacerbates the coordination problem.  Incorporating coupled constraints can enable efficient coordination between players and guarantee a certain target level of  social welfare. } This paper aims at addressing  this problem and proposes a first approach to bridge the gap between correlation and coupled constraints in games. 
 

The main contributions of this paper are as follows. 
We define a model of game with correlations and constraints and introduce a solution concept called constrained correlated equilibrium.  
We show several properties of this concept including relations to (unconstrained) correlated equilibrium.  
Furthermore, in the case of constraints on probability distributions over action profiles, we show that canonical devices are sufficient to characterize the set of constrained correlated equilibrium distributions of the game and that constrained correlated equilibria exist if the feasible set of distributions is compact and convex. 
Considering mixed strategies in the extended game,  we show that pure strategies are sufficient to generate all constrained correlated equilibrium distributions, thus showing that an additional independent randomization by the players is not necessary.   
Finally, we show numerical examples to illustrate the results of this paper with a particular focus on constraints guaranteeing some level of social welfare.

The paper is organized as follows. Section \ref{sec:relatedWork} presents the related work. Section \ref{sec:background} defines the model and gives a background on (unconstrained) correlated equilibria. 
Section \ref{sec:constrainedCorrelatedEquilibrium} defines the concept of constrained correlated equilibrium and shows some properties.
Section \ref{sec:constrainedProbabilities} considers constraints on probability distributions over action profiles, shows conditions of existence of constrained correlated equilibria in this case and studies the constrained correlated equilibrium distributions of the mixed extension of the game.
Section \ref{sec:illustration} shows numerical experiments.
Section \ref{sec:conclusions} concludes.

\section{Related Work}\label{sec:relatedWork}


Correlated equilibria have been defined by Aumann in \cite{AUMANN1974} and \cite{aumann1987correlated} showing among other important results that the set of correlated equilibrium distributions is a non-empty convex polytope as well as the connection between the concept and Bayesian rationality.
A second proof of existence and a generalization to infinite games have been proposed in \cite{hart1989existence}.
\textcolor{black}{Von Stengel and Forges} \cite{von2008extensive} extended the concept to extensive games and other generalizations of the concept to games with communications between the players can be found in \cite{forges2020correlated}. 
Furthermore, other variants of the concept are presented in \cite{brandenburger1992correlated} which considers information processing errors by the players and \cite{grant2022delegation} which allows for ambiguity in the correlation device.
Refinements using additional rationality conditions have been defined in \cite{dhillon1996perfect,myerson1984acceptable}. 
Correlated equilibria in stochastic games have been considered in \cite{solan2002correlated,solan2001characterization} and applications to systems and engineering problems can be found in \cite{altman2006correlated,eyni2021regret,xiao2016resource,chen2020correlated}.
In \cite{moulin_strategically_1978}, Moulin \emph{et al.} introduced coarse correlated equilibria as a relaxation of the concept using weaker stability conditions. 
Computational properties of correlated equilibria have been studied in \cite{stein2011correlated,papadimitriou_computing_2008,jiang2015polynomial}. 
Finally, considering learning in games,
various learning dynamics have been shown to converge to correlated equilibrium distributions \cite{foster_calibrated_1997,hart_adaptive_2005,hart2000simple,stoltz_learning_2007}.

Regarding constraints in non-cooperative games and generalized Nash equilibria, \cite{debreu1952} defines the concept of social equilibrium laying the foundations for the model of abstract economy studied in \cite{arrow1954existence}\cite{mckenzie1954equilibrium} where each player's set of feasible strategies may depend on the strategies of the other players. 
In the pioneering paper \cite{rosen_existence_1965}, Rosen considers the existence and uniqueness of equilibria in n-player concave games with shared or coupled constraints. 
See \cite{facchinei2007generalized,dutang2013existence,fischer2014generalized} for surveys on this topic 
\textcolor{black}{ (also known as generalized Nash equilibrium problem) including discussions on constraints in games, existence results and historical overviews. 
Individual constraints and their connection to shared constraints have been recently studied in \cite{Braouezec2023}.
Finally, from an application perspective, games with constraints (also called generalized games) have been used to study environmental problems and pollution \cite{krawczyk2005coupled}\cite{ krawczyk2000relaxation}, energy systems and smart grids \cite{kulkarni2017games}, electricity markets \cite{contreras2004numerical}, power allocation or computation offloading in wireless systems \cite{pang2008distributed}\cite{nowak2018generalized} and congestion in communication networks \cite{alpcan2002game}.}


\section{Correlated and generalized Nash equilibria}\label{sec:background}




Consider a finite non-cooperative game in normal form $G = (\mathcal{N}, (\mathcal{A}_i)_{i \in \mathcal{N}}, (u_i)_{i \in \mathcal{N}})$ where $\mathcal{N}$ is the set of players, $\mathcal{A}_i$ is the set of actions of player $i$ and
 $u_i: \mathcal{A}=\times_{i\in\mathcal{N}}\mathcal{A}_i \rightarrow \mathbb{R}$ is player $i$'s utility function such that her utility for the action profile $\bm{a} \in \mathcal{A}$ is $u_{i}(\bm{a})$ also denoted $ u_i(a_i,\bm{a}_{-i})$.
 The mixed extension of $G$ is the game 
 $\Delta G = (\mathcal{N}, (\Delta(\mathcal{A}_i))_{i \in \mathcal{N}}, (u_i)_{i \in \mathcal{N}})$ where $\Delta(\mathcal{A}_i) = \{\bm{p}\in \mathbb{R}^{|\mathcal{A}_i|}_+ \mid \sum_{a_i\in\mathcal{A}_i}\bm{p}(a_i) = 1\}$ is the set of probability distributions on $\mathcal{A}_i$  and $u_i: \times_{i\in\mathcal{N}}\Delta(\mathcal{A}_i) \rightarrow \mathbb{R}$ 
 is player $i$'s utility function\footnote{
    For the sake of simplicity, in this article, we use $u_i$ as a symbol for the utility function of player $i$ in any game, which one is used should be clear from the context. 
}
such that for any $\bm{p} = (\bm{p}_1, \ldots, \bm{p}_n) \in \times_{i\in\mathcal{N}}\Delta(\mathcal{A}_i)$, 
 $u_{i}(\bm{p}_i,\bm{p}_{-i}) = \sum_{\bm{a}\in\mathcal{A}}\prod_{i\in\mathcal{N}}\bm{p}_i(a_i)u_i(a_i,\bm{a}_{-i})$.
The extension of $u_i$ to the domain $\Delta(\mathcal{A})$ is the function $u_i:\Delta(\mathcal{A})\rightarrow \mathbb{R}$ such that for any $\bm{p}\in\Delta(\mathcal{A})$, $u_i(\bm{p}) = \sum_{\bm{a}\in\mathcal{A}}\bm{p}(\bm{a})u_i(\bm{a})$.



\subsection{{Correlated equilibrium}}
\label{subsec:CorrelatedEquilibrium}
A correlation device \cite{forges2020correlated} is a triplet $d = (\Omega, (\mathcal{P}_i)_{i \in \mathcal{N}}, \bm{q})$ where $\Omega$ is a set of outcomes, $\mathcal{P}_i$ a partition of $\Omega$ for player $i$ and $\bm{q}$ a probability distribution over $\Omega$. 
In this work, we assume that $\Omega$ is finite.
The pair $(G,d)$ defines a finite non-cooperative game in normal form $G_d = (\mathcal{N},(\mathcal{S}_{i,d})_{i\in\mathcal{N}},(u_i)_{i\in\mathcal{N}})$, called extended game, such that a strategy for player $i$ in $G_d$ is a $\mathcal{P}_i$-measurable mapping $\alpha_{i}: \Omega \rightarrow \mathcal{A}_{i}$ (where the measurability is w.r.t. the $\sigma$-algebra induced by the partition $\mathcal{P}_i$) and
\begin{equation}
    \mathcal{S}_{i, d} = \{ \alpha_i : \Omega \rightarrow \mathcal{A}_i \mid  \alpha_i \text{ is } \mathcal{P}_i\text{-measurable} \}
\end{equation}
In $G_d$, a strategy profile $\bm{\alpha} = (\alpha_i)_{i \in \mathcal{N}}$ is called correlated strategy profile (equivalently, correlated strategy $n$-tuple \cite{aumann1987correlated}) and the set of correlated strategy profiles is $\mathcal{S}_{d} = \times_{i \in \mathcal{N}} \mathcal{S}_{i, d}$.
The utility function $u_i:\mathcal{S}_d\rightarrow \mathbb{R}$ is defined such that, for any $\bm{\alpha}\in\mathcal{S}_{d}$,
\begin{equation}
    u_i(\alpha_i,\bm{\alpha}_{-i}) 
    = 
    \sum\limits_{\omega\in\Omega}\bm{q}(\omega)u_i(\alpha_i(\omega),\bm{\alpha}_{-i}(\omega))
\end{equation}

\noindent
The probability distribution of a correlated strategy profile $\bm{\alpha}$, denoted $\bm{p}_{\bm{\alpha}}$ is such that for any $\bm{a} \in \mathcal{A}$,
\begin{equation}
    \bm{p}_{{\bm{\alpha}}}(\bm{a}) = \sum_{\omega \in \Omega} \bm{q}(\omega) \mathds{1}_{\bm{\alpha}(\omega) = \bm{a}} \label{eq:probabilityDistribution}
\end{equation}
Furthermore, we have,
\begin{equation}
    u_i(\alpha_i,\bm{\alpha}_{-i}) 
    =
    \sum_{\bm{a}\in\mathcal{A}}\bm{p}_{\bm{\alpha}}(\bm{a})u_i(\bm{a}) 
\end{equation}
A standard interpretation of $G_d$ assumes that outcomes of $\Omega$ are drawn randomly according to the probability distribution $\bm{q}$ and, given $\omega\in\Omega$, player $i$ implementing strategy $\alpha_i\in\mathcal{S}_{i,d}$ observes her partition element ${P}_i(\omega)$ in $\mathcal{P}_i$ and plays the corresponding action $\alpha_i(\omega)$ in $\mathcal{A}_i$.

\begin{definition}[Correlated equilibrium]\label{defn:CE}
    A correlated equilibrium of $G$ is a pair $(d, \bm{\alpha}^*)$ where $d$ is a correlation device and $\bm{\alpha}^*$
    is a Nash equilibrium of $G_d$.
\end{definition} 
By definition of a Nash equilibrium, a correlated strategy profile $\bm{\alpha}^*\in\mathcal{S}_d$ is a correlated equilibrium if, for any $i \in \mathcal{N}$, for any $\alpha^\prime_{i} \in \mathcal{S}_{i, d}$,
\begin{equation}\label{eq:exantecondition} 
    \sum_{\omega \in \Omega} \bm{q}(\omega) \left[u_{i}({\alpha}_i^*(\omega), \bm{\alpha}_{-i}^*(\omega)) - u_{i}(\alpha^\prime_i(\omega), \bm{\alpha}^*_{-i}(\omega))\right] \geq 0
\end{equation}
The probability distribution $\bm{p}_{\bm{\alpha}^*}$ of a correlated equilibrium $\bm{\alpha}^*$ is called a correlated equilibrium distribution.

The original formulation of correlated equilibrium can be found in \cite{AUMANN1974} 
and equivalent ones in \cite{aumann1987correlated, forges2020correlated,brandenburger1987rationalizability,fudenberg1991game}. It is worth noting the interest and originality of the formulation used in \cite{aumann1987correlated} making the correlation device implicit (partitions being defined as preimages of actions by strategies in the correlated equilibrium strategy profile).
The equilibrium condition of Definition \ref{defn:CE} can equivalently \cite{forges2020correlated} be written for any
$i \in \mathcal{N}$, for any $a_{i} \in \mathcal{A}_{i}$ and for any $\omega \in \Omega$,
\begin{align}\label{eq:expostconditon}
    \sum_{\omega^{\prime} \in P_{i}(\omega)} \bm{q}(\omega^{\prime}) \left[u_{i}({\alpha}_i^*(\omega^{\prime}), \bm{\alpha}_{-i}^*(\omega^{\prime}))) - u_{i}(a_i, \bm{\alpha}^*_{-i}(\omega^{\prime}))\right] \geq 0 
\end{align}
In other words, it is sufficient to check the stability conditions \eqref{eq:exantecondition} with deviations $\alpha_i^\prime$ such that for any $\omega \in \Omega$,
\begin{align}\label{eq:partialDeviation}
    \alpha_i^\prime(\omega^\prime) 
    = 
    \begin{cases}
        a_i, & \text{if } \omega^\prime \in P_i(\omega)\\
        \alpha_i^*(\omega^\prime),      & \text{otherwise}
    \end{cases}
\end{align}

In this paper, we consider the formulation defining a correlated equilibrium as a Nash equilibrium, in which the correlation device is explicitly specified to facilitate the definition of the concept of constrained correlated equilibria while using a generalized device including constraints and replacing Nash equilibria by generalized Nash equilibria. 


A correlation device $d = (\Omega,(\mathcal{P}_i)_{i\in\mathcal{N}},\bm{q})$ is called canonical \cite{forges2020correlated} if $\Omega = \mathcal{A}$ and for any $i \in \mathcal{N}$, $\mathcal{P}_{i} = \mathcal{P}^c_{i}$ where the partition $\mathcal{P}^c_{i}$ is generated by $\mathcal{A}_i$ (\ie for any $\bm{a} \in \Omega$ and $\bm{a}^\prime \in \Omega,\bm{a}^\prime \in P_i(\bm{a})$ iff $a_i = a^\prime_i$). Throughout the paper, we denote by $d_c = (\mathcal{A}, (\mathcal{P}^c_i)_{i\in\mathcal{N}}, \bm{q}_c)$ a canonical correlation device.

For any correlated equilibrium $(d,\bm{\alpha}^*)$, the pair $(d_c,\bm{id})$ is a canonical correlated equilibrium with  $\bm{q}_c = \bm{p}_{\bm{\alpha}^*}$ and $\bm{id}:\mathcal{A}\rightarrow \mathcal{A}$ is the identity function\footnote{The identity function $\bm{id}$ implies that the players' choice is to play the recommended action when the canonical correlation device is used. 
} \cite{aumann1987correlated}. 
Thus, it is necessary and sufficient to consider the set of canonical correlation devices to describe the set of correlated equilibrium distributions defined as the set of probability distributions $\bm{p} \in \Delta(\mathcal{A})$ such that for any player $i \in \mathcal{N}$, for any function $\beta_i : {\mathcal{A}_i} \rightarrow {\mathcal{A}_i}$,
\begin{align}\label{eq:correlatedEquilibriumDistributions}
    \sum_{\bm{a} \in \mathcal{A}} \bm{p}(\bm{a}) [u_i(\bm{a}) - u_i(\beta_i(a_i), \bm{a}_{-i})] \geq 0 
\end{align}
which is equivalent for any player $i \in \mathcal{N}$, for any $a_i, a^\prime_i \in \mathcal{A}_i$,
\begin{align}\label{eq:equivalentCorrelatedEquilibriumDistributions}
    \sum_{\bm{a}_{-i} \in \mathcal{A}_{-i}} \bm{p}(\bm{a}) [u_i(a_i, \bm{a}_{-i}) - u_i(a_i^\prime, \bm{a}_{-i})] \geq 0 
\end{align}



The set of inequalities given in \eqref{eq:equivalentCorrelatedEquilibriumDistributions} \textcolor{black}{defines} the convex polytope of correlated equilibrium distributions of $G$, denoted $\mathcal{D}$.
Furthermore, the equivalent conditional formulation of these inequalities is such that for any $i \in \mathcal{N}$, for any $a_i \in \mathcal{A}_i$ such that $\bm{p}(a_i) > 0$ and for any $a^\prime_i \in \mathcal{A}_i$,
\begin{equation}
    \sum_{\bm{a}_{-i} \in \mathcal{A}_{-i}} \bm{p}(\bm{a}_{-i} \mid a_i) [u_i(a_i, \bm{a}_{-i}) - u_i(a_i^\prime, \bm{a}_{-i})] \geq 0 \label{eq:conditionalCanonical}
\end{equation}
leads to the interpretation of a correlated equilibrium distribution as a probability distribution over action profiles that may be used to recommend actions to the players and such that no player benefits (in expectation taken over the other players' recommendations) by deviating from her recommendation. More formally, a profile $\bm{a}$ is selected with probability $\bm{p}(\bm{a})$ and each player $i$ is recommended her component $a_i$ of the action profile. A player has no incentive to unilaterally deviate and play $a_i^\prime$ when she is recommended to play $a_i$.

Conditions \eqref{eq:correlatedEquilibriumDistributions}, \eqref{eq:equivalentCorrelatedEquilibriumDistributions} and \eqref{eq:conditionalCanonical} are commonly used as equivalent definitions of correlated equilibrium distributions in the literature (as well as definitions of correlated equilibria by considering that correlated strategy profiles inducing the same distributions are equivalent). 
However, it is worth noting that a recent work \cite{BACH202012} discusses the equivalence between the canonical definitions and the original one.
\medskip

\noindent \textbf{Example.}
In this example, we consider the game called Chicken \cite{forges2020correlated}, denoted $G$, shown in Figure \ref{fig:chickenGame} such that each player can choose to play Aggressive (action '$A$') or Peaceful (action '$P$').
\begin{figure}[H]
    \centering
    \scalebox{1.2}{
    \begin{tabular}{c|c|c}
        & \multicolumn{1}{c|}{$P$}  & $A$  \\ \hline
        $P$ & \multicolumn{1}{c|}{$(8, 8)$} & $(3, 10)$  \\ \hline
        $A$ & \multicolumn{1}{c|}{$(10, 3)$} & $(0, 0)$
    \end{tabular}}
  \caption{Game of Chicken.}
  \label{fig:chickenGame}
\end{figure}
The game has two pure Nash equilibria, $({P},{A})$ with payoffs $(3,10)$, $({A},{P})$ with payoffs $(10,3)$ and a mixed Nash equilibrium $((\sfrac{3}{5}\cdot {P},\sfrac{2}{5}\cdot {A}), (\sfrac{3}{5} \cdot {P}, \sfrac{2}{5} \cdot {A}))$ with payoffs $(6, 6)$.
Assume the correlation device $d = (\Omega, (\mathcal{P}_i)_{i \in \mathcal{N}}, \bm{q})$ such that $\Omega = \{ H, M, L\}$, $\mathcal{P}_1 = \{ \{H\}, \{M, L\}\}$, $\mathcal{P}_2 = \{ \{H, M\}, \{L\}\}$ and $\bm{q}(H) = \bm{q}(M) = \bm{q}(L) = \sfrac{1}{3}$.
Following the interpretation in terms of observations of the partition elements, in $G_d$, Player 1 observes if the outcome is $H$ or $H^c = \{M, L\}$ and player 2 observes if it is $L$ or $L^c = \{H, M\}$.
Figure \ref{fig:extendedChickenGame} shows the strategies and utilities in $G_d$ where the notation $H\mapsto P$ means that the corresponding strategy maps $H$ to $P$. 
\begin{figure}[H]
    \centering
    \scalebox{0.9}{
    \begin{tabular}{c|c|c|c|c}
                & $\bm{s}_{2}^1$ : \makecell{${L\textcolor{white}{^c} \mapsto P}$\\${L^c\mapsto P}$}
                & $\bm{s}_{2}^2$ : \makecell{${L\textcolor{white}{^c} \mapsto A}$\\${L^c \mapsto A}$}
                & $\bm{s}_{2}^3$ : \makecell{${L\textcolor{white}{^c} \mapsto A}$\\${L^c \mapsto P}$}
                & $\bm{s}_{2}^4$ : \makecell{${L\textcolor{white}{^c} \mapsto P}$\\${L^c \mapsto A}$}\\ \hline
        $\bm{s}_{1}^1$ : \makecell{${H\textcolor{white}{^c} \mapsto P}$\\${H^c \mapsto P}$} &       8, 8            &           \cellcolor{green!15}3, 10            &   6.33, 8.67        &      4.67, 9.33       \\ \hline
        $\bm{s}_{1}^2 $ :\makecell{${H\textcolor{white}{^c} \mapsto A}$\\${H^c \mapsto A}$} &   \cellcolor{green!15}10, 3           &           0, 0           &       6.67, 2        &     3.33, 1 \\ \hline
        $\bm{s}_{1}^3 $ :\makecell{${H\textcolor{white}{^c} \mapsto A}$\\${H^c \mapsto P}$} &   8.67, 6.33        &       2, 6.67                &       \cellcolor{green!15}7, 7             &       3.67, 6        \\ \hline
        $\bm{s}_{1}^4 $ :\makecell{${H\textcolor{white}{^c} \mapsto P}$\\${H^c \mapsto A}$}  &   9.33, 4.67        &       1, 3.33               &     6, 3.67          &     4.33, 4.33        \\
    \end{tabular}}
  \caption{Extension $G_d$ of the game of Chicken by the correlation device $d$.}
  \label{fig:extendedChickenGame}
\end{figure}
$G_d$ has three (pure) Nash equilibria (shown in green in Figure  \ref{fig:extendedChickenGame}), $(s_1^2,s_2^1)$, $(s_1^1,s_2^2)$ and $(s_1^3,s_2^3)$.
By definition, these are correlated equilibria of $G$, each inducing a correlated equilibrium distribution. 
In the first two equilibria, each player plays a constant function inducing correlated equilibrium distributions such that $(P,A)$ or $(A,P)$ are played with probability one, each corresponding to the pure Nash equilibria $(P,A)$ and $(A,P)$ of $G$. 
\begin{figure}[]
  \begin{subfigure}[t]{0.5\textwidth}
    \centering
    \includegraphics[scale=0.087]{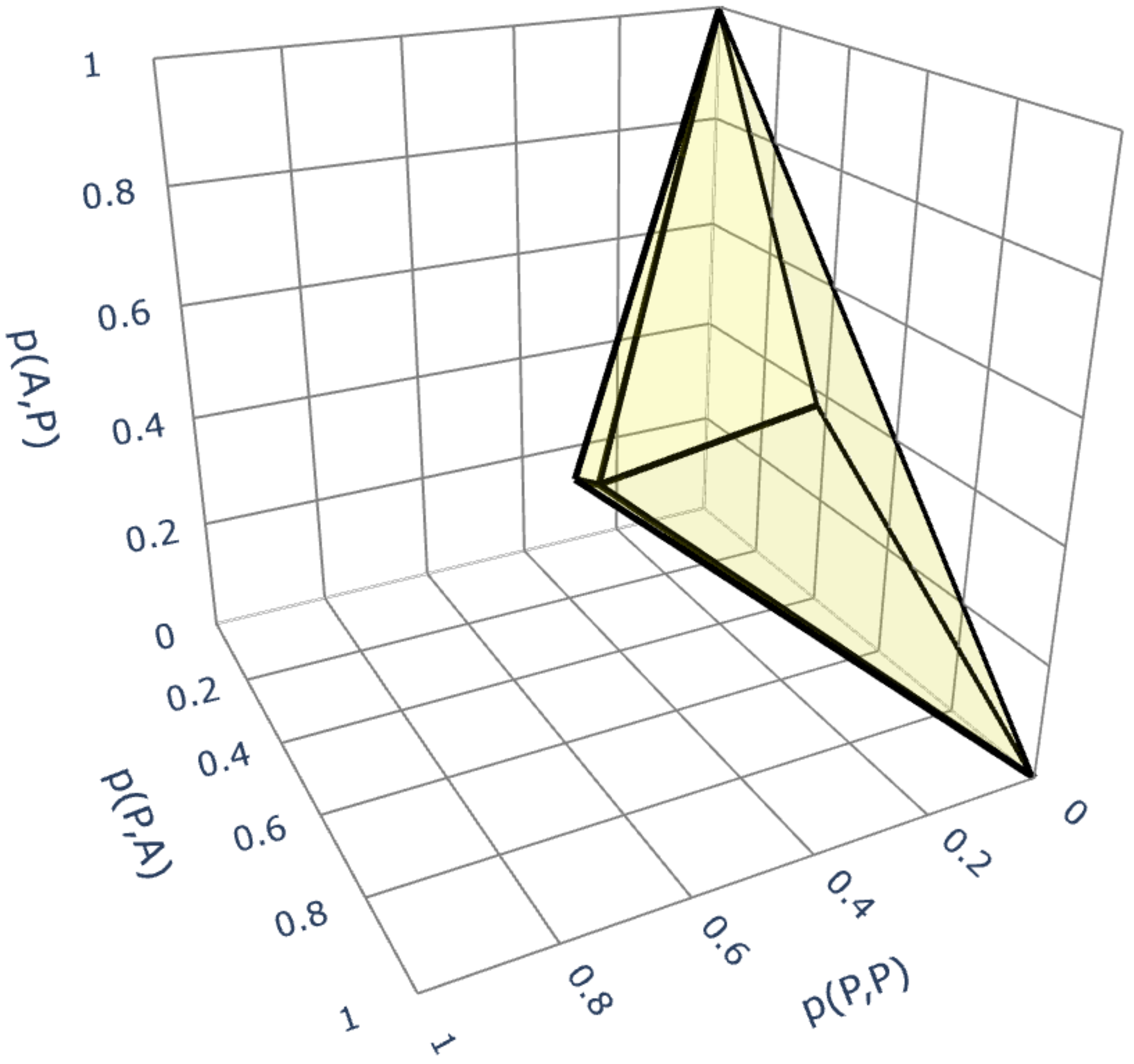}
    \caption{Probability distributions}
    \vspace{4ex}
  \end{subfigure}
  \begin{subfigure}[t]{0.5\textwidth}
    \centering
    \includegraphics[scale=0.34]{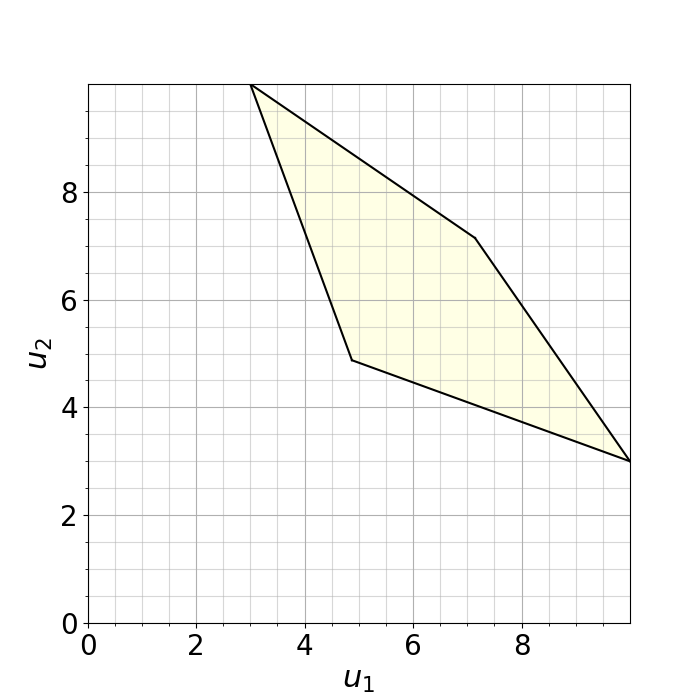}
    \caption{Utilities}
  \end{subfigure}
  \caption{(a) Polytope of correlated equilibrium distributions $\mathcal{D}$ of the game of Chicken and (b) corresponding set of pairs of utilities.}
  \label{fig:polytope}
\end{figure}
The set of correlated equilibrium distributions of $G$ is the convex polytope $\mathcal{D}$ defined by the following set of inequalities,
\begin{equation}\label{eq:polytopeofCEChicken}
    \begin{cases}
        &3\bm{p}(P,A)\geq 2\bm{p}(P,P)\\
        &2\bm{p}(A,P)\geq 3\bm{p}(P,A)\\ 
        &3\bm{p}(A,P)\geq 2\bm{p}(P,P)\\
        &2\bm{p}(P,A)\geq 3\bm{p}(A,A)\\
        &\bm{p}(A,A)+\bm{p}(P,A)+\bm{p}(A,P)+\bm{p}(P,P) = 1\\
        &\bm{p}(A,A)\geq 0, \bm{p}(P,A)\geq 0, \bm{p}(A,P)\geq, \bm{p}(P,P)\geq 0
    \end{cases}
\end{equation}
where $\bm{p}(X,Y)$ is a notation for the probability of the action profile $(X,Y)$.
Figure \ref{fig:polytope} (a) shows the set of correlated equilibrium distributions $\mathcal{D}$ and the corresponding set of utilities, as  shown in Figure \ref{fig:polytope} (b).
\textcolor{black}{Furthermore, the set of correlated equilibria $\mathcal{D}$ contains Nash equilibria that all lie on its boundary \cite{nau_geometry_2004}.
In terms of utilities, the set of correlated equilibrium utilities contains the convex hull of the set of Nash equilibrium utilities.}\\

\textcolor{black}{To conclude this section, note that correlation in games and correlated equilibria have been defined, analyzed and developed using  extended games introducing correlation as a strategic opportunity for the players (allowing them to map outcomes of the sample space to pure strategies of the original game).
A well-known result \cite{aumann1987correlated} showing that it is sufficient to focus on the so-called canonical devices has \textcolor{black}{drawn} attention to correlated equilibrium distributions, often called correlated equilibria in the literature.
In this paper, \textcolor{black}{due to the strategic aspect inherent to the concept of correlated equilibrium}, we adopt the original perspective (as done in works generalizing correlated equilibria such as \cite{brandenburger1992correlated, von2008extensive}) considering extended games and correlated equilibria as strategy profile.  
}

\subsection{{Generalized Nash equilibrium}}

A generalized game \cite{fischer2014generalized,dutang2013existence,facchinei2007generalized}, 
is a tuple 
$ 
    G' 
    = 
    (
        \mathcal{N},
        (\mathcal{S}_i)_{i\in \mathcal{N}},
        (\mathcal{K}_i)_{i\in \mathcal{N}},
        (u_i)_{i\in \mathcal{N}})
$
where 
$\mathcal{N}$, 
$\mathcal{S}_i$ and 
$u_i:\mathcal{S}\rightarrow\mathbb{R}$ are the standard components of a non-cooperative game in normal form and $\mathcal{K}_i: \mathcal{S}_{-i} \rightarrow 2^{\mathcal{S}_i}$ is a function, called constraint correspondence, such that for any $\bm{x}_{-i} \in \mathcal{S}_{-i}$, $\mathcal{K}_i(\bm{x}_{-i})\subseteq \mathcal{S}_i$ is the set of feasible strategies of player $i$.
The constraint correspondence defines player $i$'s set of feasible strategies for each profile $\bm{x}_{-i} \in \mathcal{S}_{-i}$ played by the other players.
If the strategy of player $i$ does not depend on the profile, then $\mathcal{K}_i(\bm{x}_{-i})=\mathcal{S}_i$. 

\noindent 
In this work, we assume that there is a subset of strategy profiles $\mathcal{R}\subseteq \mathcal{S}$, usually called coupled constraint set \cite{rosen_existence_1965}, such that,
\begin{equation}
    \mathcal{K}_i(\bm{x}_{-i})
    =
    \{x_i\in\mathcal{S}_i \mid (x_i,\bm{x}_{-i})\in\mathcal{R}\}
    \label{eq:constraintCorrespondance}
\end{equation}
Note that $\bm{x} \in \mathcal{R}$ if and only if for any $i\in\mathcal{N}$, $x_i \in \mathcal{K}_i(\bm{x}_{-i})$.
\begin{definition}[Generalized Nash equilibrium]\label{defn:GNE}
    A generalized Nash equilibrium of the generalized game $G'$ is a strategy profile $\bm{x}^{*}\in\mathcal{S}$ such that, for any $i \in \mathcal{N}$,
    \begin{equation}
        x_i^{*} \in 
        \underset{x_i \in \mathcal{K}_i(\bm{x}_{-i}^{*})}{\arg \max } u_i(x_i, \bm{x}_{-i}^{*}) 
    \end{equation}
\end{definition}
Equivalently, 
$\bm{x}^*$ is a generalized Nash equilibrium if and only if $\bm{x}^{*} \in \mathcal{R}$ and for any $i\in \mathcal{N}$, 
for any $x^\prime_i \in \mathcal{S}_i$ such that $(x^\prime_i, \bm{x}^*_{-i}) \in \mathcal{R}$,
    \begin{equation}\label{eq:GNE}
            u_i(x^*_i, \bm{x}_{-i}^{*}) \geq u_i(x^\prime_i, \bm{x}_{-i}^{*})
    \end{equation}



Studies on generalized games typically assume that $\mathcal{R}$ is a continuous or convex subset of a Euclidean space $\mathbb{R}^m$ for some $m \in \mathbb{N}$.
In this work, we consider only finite games such that for any player $i$, $\mathcal{S}_i$ is finite, implying the finiteness of $\mathcal{R}$.

\textcolor{black}{
It may be argued that when playing the game, players can choose strategies resulting in an unfeasible profile. 
This problem and related ones have been studied in the literature (see \cite{Braouezec2023} and references therein) and typical answers involve agreements, self-restrictions or external enforcement such as regulation. 
In this paper, we do not consider this problem and follow the standard perspective on generalized games, defining the model (including constraints), identifying a relevant solution concept (taking constraints into account) and computing or characterizing the strategy profiles satisfying the equilibrium conditions.  
}

\section{Constrained correlated equilibrium}\label{sec:constrainedCorrelatedEquilibrium}


\textcolor{black}{
This section defines the concept of constrained correlated equilibrium with respect to a correlation device $d$. 
Particularly, we introduce constraints in the extended game $G_d$ (see Section \ref{subsec:CorrelatedEquilibrium}) and study the equilibrium conditions as well as some consequences. 
These results will be used in Section \ref{sec:constrainedProbabilities} focusing on the equilibria of a collection of extended games with constraints induced by a feasible set of probability distributions over action profiles, which is particularly relevant with respect to systems and applications.}

Let $G=(\mathcal{N},(\mathcal{S}_{i})_{i\in\mathcal{N}}, (u_i)_{i\in\mathcal{N}})$ be a non-cooperative game in normal form and let $G_d = (\mathcal{N},(\mathcal{S}_{i,d})_{i\in\mathcal{N}}, (u_i)_{i\in\mathcal{N}})$ be its extension by a correlation device $d$. 
Furthermore, let $\mathcal{R}_d \subseteq \mathcal{S}_{d}$ be a coupled constraint set in $G_d$ inducing the generalized game
$ 
    G_d' 
    = 
    (
        \mathcal{N},
        (\mathcal{S}_{i,d})_{i\in \mathcal{N}},
        (\mathcal{K}_{i,d})_{i\in \mathcal{N}},
        (u_i)_{i\in \mathcal{N}})
$
as defined in Section \ref{sec:background}.
By definition of the constraint correspondences with coupled constraint, for any correlated strategy profile $\bm{\alpha}_{-i}\in\mathcal{S}_{-i, d}$, the set of feasible strategies of player $i$ is,
\begin{equation}\label{eq:definitionKi}
    \mathcal{K}_{i,d}(\bm{\alpha}_{-i}) 
    =
    \{\alpha_i \in \mathcal{S}_{i,d} \mid (\alpha_i , \bm{\alpha}_{-i}) \in \mathcal{R}_d \} 
\end{equation}
As in the case without constraints, in the extended game with coupled constraint $G_d'$, the probability distribution of a correlated strategy profile $\bm{\alpha}$ is denoted $\bm{p}_{\bm{\alpha}}$.
We define a constrained correlated equilibrium of $G$ as a generalized equilibrium of an extension $G_d$ with corresponding coupled constraint $\mathcal{R}_d$ such that,
\begin{definition}[Constrained Correlated Equilibrium]\label{defn:constrainedCE}
            A constrained correlated equilibrium of $G$ is a triplet $(d, \mathcal{R}_d, \bm{\alpha}^*)$ where $d$ is a correlation device, $\mathcal{R}_d$ is a coupled constraint set in $G_d$ and $\bm{\alpha}^*$ is a correlated strategy profile such that, $\bm{\alpha}^*\in \mathcal{R}_d$ and 
            for any $i\in\mathcal{N}$, 
            for any $\alpha_i^\prime \in \mathcal{S}_{i,d}$ such that $(\alpha_i^\prime, \bm{\alpha}^*_{-i}) \in \mathcal{R}_d$,
            \begin{align}\label{eq:constrainedCorrelatedEquilibriumConditions}
                \sum\limits_{\omega \in \Omega} \bm{q}\left(\omega\right) 
                \left[ 
                    u_{i}\left({\alpha}_i^*\left(\omega\right), \bm{\alpha}_{-i}^*\left(\omega\right)\right) 
                    - 
                    u_{i}\left(\alpha^\prime_{i}(\omega), \bm{\alpha}^*_{-i}(\omega)\right) 
                \right] \geq 0
            \end{align}
        \end{definition}

It can be easily checked that $(d, \mathcal{R}_d, \bm{\alpha}^*)$ is a constrained  correlated equilibrium  if only if $\bm{\alpha}^*$ is a generalized Nash equilibrium of $G^\prime_d$, \ie
for any $i \in \mathcal{N}$,
\begin{equation}
    \alpha_i^{*} \in 
    \underset{\alpha_i \in \mathcal{K}_{i,d}(\bm{\alpha}_{-i}^{*})}{\arg \max } u_i(\alpha_i, \bm{\alpha}_{-i}^{*}).
\end{equation}



The probability distribution of a  constrained correlated equilibrium with strategy profile $\bm{\alpha}^*$ is denoted $\bm{p}_{\bm{\alpha}^*}$ and called constrained correlated equilibrium distribution of $G$.
Furthermore, if for some device $d$, $\mathcal{R}_d = \mathcal{S}_d$, then, for any player $i$, $\mathcal{K}_{i,d}(\bm{\alpha}_{-i}) = \mathcal{S}_{i,d}$ and a constrained correlated equilibrium strategy profile (w.r.t. $d$) is a correlated equilibrium strategy profile.

In spite of its relevance, we do not use the term generalized correlated equilibrium as it is used to refer to the solution concept developed in \cite{brandenburger1992correlated} considering general information structures and in \cite{forgo2010generalization} to refer to a correlation scheme generalizing correlated equilibria in finite games.

The following example illustrates the concept of constrained correlated equilibrium.
\medskip

\noindent \textbf{Example.} 
Consider two drivers (players) arriving at an intersection. 
Each can either cross the intersection or wait. 
The utility for waiting is 0 (whatever the action played by the other player), if both cross a collision occurs, each player having utility $-1$ for this outcome and if only one player goes, she has utility $+1$.
Figure \ref{fig:intersectionGame} shows the corresponding traffic intersection game \cite{brandenburg2022combinatorics} \cite{owen2013game}.
This game has two pure Nash equilibria $(Wait, Go)$ and $(Go, Wait)$ and one mixed Nash equilibrium $((\sfrac{1}{2}\cdot Wait, \sfrac{1}{2}\cdot Go), (\sfrac{1}{2} \cdot Wait, \sfrac{1}{2} \cdot Go))$.
\begin{figure}[H]
    \centering
    \scalebox{1.2}{
        \begin{tabular}{c|c|c}
                 & \multicolumn{1}{c|}{$Wait$}  & $Go$  \\ \hline
            $Wait$ & \multicolumn{1}{c|}{$(0, 0)$} & $(0, 1)$  \\ \hline
            $Go$ & \multicolumn{1}{c|}{$(1, 0)$} & $(-1, -1)$
        \end{tabular}}
    \caption{Intersection game.}
    \label{fig:intersectionGame}
\end{figure}
Assume traffic lights sending a signal in $\{green,red\}$ to each player. 
For the sake of simplicity, these signals are denoted $g$ and $r$ and $(g,r)$ is the pair of signals such that the first player observes $green$ and the second one $red$. 
We assume traffic lights such that the pairs $(green,green)$ and $(red,red)$ do not occur.
Furthermore, assume that players must follow driving rules such that if the received signal is $green$ the player must cross, otherwise, she waits.\\
Following \cite{brandenburg2022combinatorics}, such traffic lights can be modelled by a correlation device $d = (\Omega, (\mathcal{P}_i)_{i \in \mathcal{N}}, \bm{q})$ 
such that the set of outcomes is \textcolor{black}{$\Omega = \{ (r,g), (g,r)\}$, $\mathcal{P}_1 = \{ \{(r,g)\}, \{(g,r) \}\}$, $\mathcal{P}_2 = \{ \{(r,g)\}, \{(g,r)\}\}$ and 
$\bm{q} = (\sfrac{1}{2}, \sfrac{1}{2})$.}
Furthermore, given the extension of the traffic intersection game by $d$\footnote{
    By definition, $\mathcal{S}_d = \mathcal{S}_{1,d}\times \mathcal{S}_{2,d}$ is the set of pairs of functions from $\Omega=\{(r,g),(g,r)\}$ to $\{Go,Wait\}$, each measurable w.r.t. the corresponding partition.
}, driving rules \textcolor{black}{can be modelled by} a coupled constraint set $\mathcal{R}_d$ such that,

\begin{equation}\label{eq:Specific_behavior_traffic_lights}
    \mathcal{R}_{d} 
    = 
    \left\{  
        (\alpha_1, \alpha_2) \in \mathcal{S}_{d} 
        \mid  
        \alpha_1(r,g) = \alpha_2(g,r) = Wait, \alpha_1(g,r) = \alpha_2(r,g) = Go 
    \right\}
\end{equation}
where $\alpha_i(X, Y)$ is a simplified notation for $\alpha_i((X, Y))$.
\textcolor{black}{This constraint set} being a singleton, the triplet $(d,\mathcal{R}_{d},\bm{\alpha}^*)$ such that $\bm{\alpha}^*(r,g) = (Wait,Go)$ and $\bm{\alpha}^*(g,r) = (Go,Wait)$ is the only constrained correlated equilibrium for $\mathcal{R}_{d}$ \textcolor{black}{($\bm{\alpha}^*$ is feasible and no unilateral deviation induces a feasible profile)}.
\textcolor{black}{
Now, assume that the objective of the constraint is not to enforce a given behavior but to avoid collisions (\ie models relaxed driving rules such that "drivers must self-organize w.r.t. signals but cannot both cross or stop").
Then, the coupled constraint set \eqref{eq:Specific_behavior_traffic_lights} is not necessary anymore and it is sufficient to require that the players choose a correlated strategy profile $\bm{\alpha}$ such that the induced distribution $\bm{p}_{\bm{\alpha}}$ satisfies $\bm{p}_{\bm{\alpha}}(Go,Go)=0$ and $\bm{p}_{\bm{\alpha}}(Wait,Wait)=0$,
}
\textcolor{black}{
\begin{equation}\label{eq:No_collision_behavior_traffic_lights}
    \mathcal{R}'_{d} 
    = 
    \left\{  
        \bm{\alpha}\in \mathcal{S}_{d} 
        \mid  
        \bm{p}_{\bm{\alpha}}(Go,Go)=\bm{p}_{\bm{\alpha}}(Wait,Wait)=0
    \right\}
\end{equation}
}
\textcolor{black}{This coupled constraints set contains a second constrained correlated equilibrium strategy profile $\bm{\delta}^*$ such that $\bm{\delta}^*(r,g) = (Go,Wait)$ and $\bm{\delta}^*(g,r) = (Wait,Go)$, \ie ("Go if red, Wait if green").} 
Both constrained equilibrium strategy profiles $\bm{\alpha}^*$ and $\bm{\delta}^*$ induce the same probability distribution such that $\bm{p}_{\bm{\alpha}^*}(Go,Wait) = \bm{p}_{\bm{\alpha}^*}(Wait,Go) = \sfrac{1}{2}$. We observe that $\bm{\alpha}^*$ is the strategy profile of the constrained correlated equilibria $(d,\mathcal{R}^\prime_{d},\bm{\alpha}^*)$ and $(d,\mathcal{R}_{d}, \bm{\alpha}^*)$ where $\mathcal{R}_{d} \subseteq \mathcal{R}^\prime_{d}$. 
The next proposition shows that this result holds for any finite non-cooperative game.



\begin{proposition}\label{prop:robustnessToRestrictions}
    Let $G$ be a finite non-cooperative game, $d$ a correlation device. 
    Furthermore, let $\mathcal{R}_d$ and $\mathcal{R}^\prime_d$ be coupled constraint sets such that $\mathcal{R}_d\subseteq \mathcal{R}_d^\prime$.
    If $(d,\mathcal{R}^\prime_d,\bm{\alpha}^*)$ is a constrained correlated equilibrium and $\bm{\alpha}^* \in \mathcal{R}_d$, then $(d,\mathcal{R}_d,\bm{\alpha}^*)$ is a constrained correlated equilibrium.
\end{proposition}
\begin{proof}[\textbf{Proof of Proposition \ref{prop:robustnessToRestrictions}}]
We prove the result by contradiction. Assume $\mathcal{R}_d$ and $\mathcal{R}^\prime_{d}$ are two feasible sets of correlated strategies such that $\mathcal{R}_d \subseteq \mathcal{R}^\prime_d$ and $(d,\mathcal{R}^\prime_d,\bm{\alpha}^*)$ is a constrained correlated equilibrium. Suppose that $\bm{\alpha}^* \in \mathcal{R}_d$ and $(d,\mathcal{R}_d,\bm{\alpha}^*)$ is not a constrained correlated equilibrium.\\
Since $(d, \mathcal{R}^\prime_d, \bm{\alpha}^*)$ is a constrained correlated equilibrium, we have for any $i \in \mathcal{N}$, for any $\alpha_i^\prime \in \mathcal{K}^\prime_{i}(\bm{\alpha}^*_{-i})$,
    \begin{equation} \label{eq:stabilityCondition}
        \sum\limits_{\omega \in \Omega} \bm{q}\left(\omega\right) 
    \left[ 
        u_{i}\left({\alpha}_i^*\left(\omega\right), \bm{\alpha}_{-i}^*\left(\omega\right)\right) 
        - 
        u_{i}\left(\alpha^\prime_{i}(\omega), \bm{\alpha}^*_{-i}(\omega)\right) 
    \right] \geq 0
    \end{equation}
Let $\alpha_i^\prime \in \mathcal{K}_{i,d}(\bm{\alpha}^*_{-i})$. By definition of the constraint correspondence $\mathcal{K}_{i,d}$, $(\alpha^\prime_i, \bm{\alpha}^*_{-i}) \in \mathcal{R}_d$. Since, $\mathcal{R}_d \subseteq \mathcal{R}^\prime_d$, we also have $(\alpha_i^\prime, \bm{\alpha}^*_{-i}) \in \mathcal{R}^\prime_d$. Hence, $\alpha^\prime_i \in \mathcal{K}_{i,d}(\bm{\alpha}^*_{-i})$ implies $\alpha^\prime_i \in \mathcal{K}_{i,d}^\prime(\bm{\alpha}^*_{-i})$.
Equation \eqref{eq:stabilityCondition} then becomes for any $i \in \mathcal{N}$, for any $\alpha_i^\prime \in \mathcal{K}_{i,d}(\bm{\alpha}^*_{-i})$,
    \begin{equation}
    \sum\limits_{\omega \in \Omega} \bm{q}\left(\omega\right) 
    \left[ 
        u_{i}\left({\alpha}_i^*\left(\omega\right), \bm{\alpha}_{-i}^*\left(\omega\right)\right) 
        - 
        u_{i}\left(\alpha^\prime_{i}(\omega), \bm{\alpha}^*_{-i}(\omega)\right) 
    \right] \geq 0
    \end{equation}
Furthermore, by assumption, $\bm{\alpha}^* \in \mathcal{R}_d$. Thus, the triplet $(d, \mathcal{R}_d, \bm{\alpha}^*)$ is a constrained correlated equilibrium, which is a contradiction.
\end{proof}
\textcolor{black}{As a consequence, the following corollary shows that any correlated equilibrium satisfying the constraints is a constrained correlated equilibrium.}



\begin{corollary}
\label{prop:CEinCCE}
    Let $G$ be a finite non-cooperative game, $d$ a correlation device and $\mathcal{R}_d$ a coupled constraint set.
    Let $(d,\bm{\alpha}^*)$ be a correlated equilibrium. 
    If $\bm{\alpha}^* \in \mathcal{R}_d$, then $(d,\mathcal{R}_d,\bm{\alpha}^*)$ is a constrained correlated equilibrium.
\end{corollary}
\begin{proof}[\textbf{Proof of Corollary \ref{prop:CEinCCE}}]
Since $\mathcal{R}_d \subseteq  \mathcal{S}_{d}$, according to 
proposition~\ref{prop:robustnessToRestrictions}, the proof is immediate. 
\end{proof}
\textcolor{black}{
The following proposition is an alternative characterization of constrained correlated equilibria showing that for any pair $(d,\mathcal{R}_d)$, the correlated strategy profile $\bm{\alpha}^*$ is a constrained correlated equilibrium if, in $G_d$, any unilateral deviation either decreases the player's utility or induces an unfeasible correlated strategy profile.
}
\begin{proposition}\label{prop:alternativeDefinition}
    Let 
    ${G}=(\mathcal{N},(\mathcal{A}_{i})_{i \in \mathcal{N}},(u_{i})_{i \in \mathcal{N}})$ 
    be a finite non-cooperative game, 
    $d$ a correlation device and $\mathcal{R}_d$ a coupled constraint set.
    The triplet $(d,\mathcal{R}_d,\bm{\alpha}^*)$ is a constrained correlated equilibrium if and only if $\bm{\alpha}^* \in \mathcal{R}_d$ and for any 
    $i \in \mathcal{N}$, for any $\alpha^\prime_{i} \in \mathcal{S}_{i,d}$,
    \begin{equation}\label{defn:constrainedCEAlternative1}
        \sum 
        \limits_{\omega \in \Omega} 
        \bm{q}(\omega) 
        \left[ 
        u_{i}({\alpha}_i^*(\omega), \bm{\alpha}_{-i}^*(\omega)) - u_{i}(\alpha^\prime_{i}(\omega), \bm{\alpha}^*_{-i}(\omega)) 
        \right] 
        \geq 0  
        \; \text{ or } \;  (\alpha^\prime_{i}, \bm{\alpha}^*_{-i}) \notin \mathcal{R}_d
    \end{equation}
\end{proposition}

\begin{proof}[\textbf{Proof of Proposition \ref{prop:alternativeDefinition}}]
    ($\Rightarrow$) 
    If $(d,\mathcal{R}_d,\bm{\alpha}^*)$ is a constrained correlated equilibrium then $\bm{\alpha}^* \in \mathcal{R}_d$ and for any $i \in \mathcal{N}$, 
    for any
    $\alpha^\prime_{i} \in \mathcal{K}_{i,d}(\bm{\alpha}^*_{-i})$,
    \begin{equation}
        \sum\limits_{\omega \in \Omega} \bm{q}(\omega)
        \left[u_{i}({\alpha}_i^*(\omega), \bm{\alpha}_{-i}^*(\omega)) - u_{i}(\alpha^\prime_{i}(\omega), \bm{\alpha}^*_{-i}(\omega)) \right] \geq 0
    \end{equation}
    Furthermore, if $\alpha_i^\prime \notin \mathcal{K}_{i,d}(\bm{\alpha}^*_{-i})$ then $(\alpha^\prime_i, \bm{\alpha}^*_{-i}) \notin \mathcal{R}_{d}$. 
    Thus, for any 
    $i \in \mathcal{N}$, 
    $\alpha^*_i\in\mathcal{K}_{i,d}(\bm{\alpha}^*_{-i})$ 
    and for any    
    $\alpha^\prime_{i} \in \mathcal{S}_{i,d}$,
    \begin{equation}
        \sum 
        \limits_{\omega \in \Omega} \bm{q}(\omega) 
        \left[ 
        u_{i}({\alpha}_i^*(\omega), \bm{\alpha}_{-i}^*(\omega)) - u_{i}(\alpha^\prime_{i}(\omega), \bm{\alpha}^*_{-i}(\omega)) 
        \right] 
        \geq 0  
        \; \text{ or }  \;  
        (\alpha^\prime_{i}, \bm{\alpha}^*_{-i}) \notin \mathcal{R}_d 
    \end{equation}
    \noindent ($\Leftarrow$) 
    Let $\bm{\alpha}^*\in\mathcal{S}_{i,d}$ such that $\bm{\alpha}^*\in \mathcal{R}_d$ 
    and for any
    $i \in \mathcal{N}$, 
    for any $\alpha^\prime_{i} \in \mathcal{S}_{i,d}$,
    \begin{equation}
        \sum 
        \limits_{\omega \in \Omega} \bm{q}(\omega) 
        \left[ 
        u_{i}({\alpha}_i^*(\omega), \bm{\alpha}_{-i}^*(\omega)) - u_{i}(\alpha^\prime_{i}(\omega), \bm{\alpha}^*_{-i}(\omega)) 
        \right] 
        \geq 0  
        \; \text{ or }  \;  
        (\alpha^\prime_{i}, \bm{\alpha}^*_{-i}) \notin \mathcal{R}_d \label{eq:1} 
    \end{equation}
    Furthermore, let $\alpha_i^\prime\in\mathcal{S}_{i,d}$ 
    such that
    $\alpha_i^\prime\in\mathcal{K}_{i,d}(\bm{\alpha}^*_{-i})$. 
    Then 
    $(\alpha_i^\prime, \bm{\alpha}^*_{-i}) \in \mathcal{R}_d$ implying,
    \begin{equation}
         \sum 
         \limits_{\omega \in \Omega} 
         \bm{q}(\omega) 
         \left[ 
         u_{i}({\alpha}_i^*(\omega), \bm{\alpha}_{-i}^*(\omega)) - u_{i}(\alpha^\prime_{i}(\omega), \bm{\alpha}^*_{-i}(\omega)) 
         \right] 
         \geq 0
    \end{equation}
    Therefore, $(d,\mathcal{R}_d,\bm{\alpha}^*)$ is a constrained correlated equilibrium.
\end{proof}

Remark that Definition \ref{defn:constrainedCE} and Proposition \ref{prop:alternativeDefinition} define the same set of strategy profiles but do not make use of the same set of strategies. 
The former uses feasible strategies in $\mathcal{R}_d$ while the latter uses all strategies in $\mathcal{S}_d$. 
This does not make a difference if knowing $\mathcal{S}_d$ but the second formulation cannot be used if knowing only the coupled constraint set or constraint correspondences. 
Despite of its interest, a detailed discussion on this problem is beyond the scope of the paper.

The following proposition gives sufficient conditions of equilibrium based on the optimality of the strategy at each outcome.
\begin{proposition}\label{prop:sufficientConditions}
    Let ${G}=(\mathcal{N},(\mathcal{A}_{i})_{i \in \mathcal{N}},(u_{i})_{i \in \mathcal{N}})$ be a finite non-cooperative game, 
    $d$ a correlation device and $\mathcal{R}_d$ a coupled constraint set. 
    If $\bm{\alpha}^* \in \mathcal{R}_d$ and for any $i \in \mathcal{N}$, for any $\alpha_i^\prime \in \mathcal{K}_{i,d}(\bm{\alpha}^*_{-i})$, for any $\omega \in \Omega$
    \begin{equation}\label{eq:perOutcomeStabilityCondition}
        \sum_{\omega^{\prime} \in P_{i}(\omega)} \bm{q}(\omega^{\prime}) \left[u_{i}({\alpha}_i^*(\omega), \bm{\alpha}_{-i}^*(\omega)) - u_{i}(\alpha_i^\prime(\omega^\prime), \bm{\alpha}^*_{-i}(\omega^{\prime}))\right] \geq 0
    \end{equation}
    then $(d,\mathcal{R}_d,\bm{\alpha}^*)$ is a constrained correlated equilibrium.
\end{proposition}
\begin{proof}[\textbf{Proof of Proposition \ref{prop:sufficientConditions}}]
    Let $\bm{\alpha}^* \in \mathcal{R}_d$ be a correlated strategy profile such that for any 
    $i \in \mathcal{N}$,
    for any $\alpha_i^\prime \in \mathcal{K}_{i,d}(\bm{\alpha}^*_{-i})$,
    for any $\omega \in \Omega$,
    \begin{equation}
        \sum_{\omega^{\prime} \in P_{i}(\omega)} \bm{q}(\omega^{\prime}) \left[u_{i}({\alpha}_i^*(\omega^\prime), \bm{\alpha}_{-i}^*(\omega^\prime)) - u_{i}(\alpha_i^\prime(\omega^\prime), \bm{\alpha}^*_{-i}(\omega^{\prime}))\right] \geq 0
    \end{equation}
    Let $\mathcal{P}_i = \{ P_{i, 1}, ..., P_{i, n_i} \}$ be player $i$'s partition of $\Omega$ in $d$ and $\Omega_i$ be a set of outcomes such that for any $k\in\{1,\ldots,n_i\}$,  $|\Omega_i \cap  P_{i,k}| = 1$ ($\Omega_i$ contains one element per element of partition of $\Omega$). 
    We have, for any 
    $i \in \mathcal{N}$,
    for any $\alpha_i^\prime \in \mathcal{K}_{i,d}(\bm{\alpha}^*_{-i})$,    
    \begin{equation}
        \sum_{\omega \in \Omega_i} 
        \sum_{\omega^{\prime} \in P_{i}(\omega)} 
        \bm{q}(\omega^{\prime}) 
        \left[u_{i}({\alpha}_i^*(\omega^\prime), \bm{\alpha}_{-i}^*(\omega^\prime)) - u_{i}(\alpha_i^\prime(\omega^\prime), \bm{\alpha}^*_{-i}(\omega^{\prime}))\right] \geq 0
    \end{equation}
    Equivalently, for any player $i \in \mathcal{N}$, for any $\alpha_i^\prime \in \mathcal{K}_{i,d}(\bm{\alpha}^*_{-i})$ 
    \begin{equation}
        \sum_{\omega \in \Omega} 
        \bm{q}(\omega) 
        \left[u_{i}({\alpha}_i^*(\omega), \bm{\alpha}_{-i}^*(\omega)) - u_{i}(\alpha_i^\prime(\omega), \bm{\alpha}^*_{-i}(\omega))\right] \geq 0
    \end{equation}
    Thus, $(d,\mathcal{R}_d,\bm{\alpha}^*)$ is a constrained correlated equilibrium.
\end{proof}

In \cite{forges2020correlated} it is shown that the stability of a correlated equilibrium strategy is induced by a subset of deviations. 
Particularly, it is sufficient for a player to consider alternative strategies deviations that are constant across all elements of her partition except one (see Equation \eqref{eq:expostconditon}). A straightforward adaptation of these equations taking constraints into account is such that 
    for player $i \in \mathcal{N}$,
    for any outcome $\omega \in \Omega$,
    for any action $a_i  \in \mathcal{A}_i$, 
    for any strategy $\alpha^\prime_{i}\in\mathcal{K}_{i,d}(\bm{\alpha}_{-i}^*)$ such that $\alpha^\prime_i(\omega^\prime)= a_i$ if  $\omega^\prime \in P_i(\omega)$ and $\alpha_i^\prime(\omega^\prime) = \alpha_i^*(\omega^\prime)$ otherwise,
    \begin{equation}\label{eq:partialStabilityConditions}
            \sum\limits_{\omega^{\prime}\in P_{i}(\omega)} 
            \bm{q}(\omega^{\prime}) 
            \left[ 
                u_{i}({\alpha}_i^*(\omega^{\prime}), \bm{\alpha}_{-i}^*(\omega^{\prime})) - u_{i}(a_i, \bm{\alpha}^*_{-i}(\omega^{\prime}))
            \right] \geq 0   
    \end{equation}
    An interesting problem is whether or not the latter conditions characterize constrained correlated equilibrium strategies. The following example shows that considering such deviations that are feasible is not sufficient to characterize constrained correlated equilibrium strategies.
\medskip

\noindent \textbf{Example.} 
Consider the extension $G_d$ of the game of Chicken shown in Figure \ref{fig:extendedChickenGame}
and assume the coupled constraint set:
$$ \mathcal{R}_d = \{(s_1^1,s_2^1), (s_1^1,s_2^2), (s_1^2,s_2^2), (s_1^2,s_2^3),(s_1^2,s_2^4), (s_1^3,s_2^3), (s_1^3,s_2^4), (s_1^4,s_2^2), (s_1^4,s_2^4)\}$$
Figure \ref{fig:ConstrainedExtendedPayoff} shows the game.
For the pair $(d,\mathcal{R}_d)$, there are three constrained equilibrium strategy profiles (in green) satisfying the equilibrium conditions \eqref{eq:constrainedCorrelatedEquilibriumConditions}.
The strategy profile $(\bm{s}_{1}^3, \bm{s}_{2}^4)$ with utilities $(3.67, 6)$ satisfies \eqref{eq:partialStabilityConditions} for every feasible deviation but is not a constrained correlated equilibrium. 
In fact, the row player can improve her utility by deviating simultaneously over the two elements of her partition to play the feasible strategy $\bm{s}_{1}^4$. This shows that conditions  \eqref{eq:partialStabilityConditions} are not sufficient to characterise a constrained correlated equilibrium.

This example also shows that the correlated strategy profile $\bm{\alpha}^* = (\bm{s}_{1}^4, \bm{s}_{2}^4)$ induces a probability distribution $\bm{p}_{\bm{\alpha}^*}$ such that $\bm{p}_{\bm{\alpha}^*}(P,P) = 0$, $\bm{p}_{\bm{\alpha}^*}(P,A) = \sfrac{1}{3}$, $\bm{p}_{\bm{\alpha}^*}(P,P) = \sfrac{1}{3}$, $\bm{p}_{\bm{\alpha}^*}(A,A) = \sfrac{1}{3}$ which is not in the polytope \eqref{eq:polytopeofCEChicken} of correlated equilibrium distributions. This implies that a constrained correlated equilibrium is not necessarily a correlated equilibrium.
\begin{figure}[H]
    \centering
    \scalebox{0.9}{
    \begin{tabular}{c|c|c|c|c}
                & $\bm{s}_{2}^1$ : \makecell{${L\textcolor{white}{^c} \mapsto P}$\\${L^c\mapsto P}$}
                & $\bm{s}_{2}^2$ : \makecell{${L\textcolor{white}{^c} \mapsto A}$\\${L^c \mapsto A}$}
                & $\bm{s}_{2}^3$ : \makecell{${L\textcolor{white}{^c} \mapsto A}$\\${L^c \mapsto P}$}
                & $\bm{s}_{2}^4$ : \makecell{${L\textcolor{white}{^c} \mapsto P}$\\${L^c \mapsto A}$}\\ \hline
    $\bm{s}_{1}^1 $ : \makecell{${H\textcolor{white}{^c} \mapsto P}$\\${H^c \mapsto P}$} &       $8, 8$            &           \cellcolor{green!15}$3, 10$            &    \cellcolor{gray!20}\textcolor{gray}{\xcancel{$6.33, 8.67$}}        &      \cellcolor{gray!20}\textcolor{gray}{\xcancel{$4.67, 9.33$}}       \\ \hline
    $\bm{s}_{1}^2 $ : \makecell{${H\textcolor{white}{^c} \mapsto A}$\\${H^c \mapsto A}$} &   \cellcolor{gray!20}\textcolor{gray}{\xcancel{$10, 3$}}           &           $0, 0$           &       $6.67, 2$        &     $3.33, 1$ \\ \hline
    $\bm{s}_{1}^3 $ : \makecell{${H\textcolor{white}{^c} \mapsto A}$\\${H^c \mapsto P}$} &   \cellcolor{gray!20}\textcolor{gray}{\xcancel{$8.67, 6.33$}}        &       \cellcolor{gray!20}\textcolor{gray}{\xcancel{$2, 6.67$}}                &       \cellcolor{green!15}$7,7$             &       $3.67, 6$        \\ \hline
    $\bm{s}_{1}^4 $ :\makecell{${H\textcolor{white}{^c} \mapsto P}$\\${H^c \mapsto A}$}  &   \cellcolor{gray!20}\textcolor{gray}{\xcancel{$9.33, 4.67$}}        &       $1, 3.33$               &     \cellcolor{gray!20}\textcolor{gray}{\xcancel{$6, 3.67$}}          &     \cellcolor{green!15}$4.33, 4.33$        \\
    \end{tabular}}
    \caption{Constrained extension of the game of Chicken.}
  \label{fig:ConstrainedExtendedPayoff}
\end{figure}

To conclude this section, remark that conditions \eqref{eq:partialStabilityConditions} are related to the concept of regret \cite{hart2000simple} used by well known learning dynamics \cite{foster_calibrated_1997, hart_regret-based_2003, hart_adaptive_2005} converging to the set of correlated equilibria\footnote{A history with no regrets being equivalent to a corresponding empirical distribution over action profiles in the polytope of correlated equilibrium distributions.}.
The latter discussion shows that learning constrained correlated equilibria may need an alternative concept of regret based on the conditions given in Proposition \ref{prop:sufficientConditions} to approach the set of constrained correlated equilibrium distributions.
\textcolor{black}{In spite of} its interest, this problem is beyond the scope of this paper.



\section{Constraints on probability distributions over action profiles}\label{sec:constrainedProbabilities}
In this section, we study the particular case of coupled constraint sets induced by a set of probability distributions over action profiles. 
In addition to their theoretical interest, such constraints are typically relevant in applications such as economics or engineering where performance criteria depend on probability distribution over action profiles but are independent from the devices. 
As an example, one may search for a constrained correlated equilibrium with constraints on the social welfare \cite{branzei2017nash} (see Section \ref{sec:illustration} for numerical experiments with a constraint on the social welfare) or the Nash product \cite{moulin2004fair} \cite{nash1950bargaining} or such that the probability of some event in $\mathcal{A}$ is below a given threshold.

Let $\mathcal{C} \subseteq \Delta(\mathcal{A})$ be a set of probability distributions, called {feasible set of probability distributions} and for each correlation device $d$, define the coupled constraint set $\mathcal{R}_d$ generated by $\mathcal{C}$ such that, 
\begin{equation}\label{eq:RdInducedByFeasibleDistributions}
    \mathcal{R}_d = \{  \bm{\alpha}  \in \mathcal{S}_d \mid \bm{p}_{\bm{\alpha}} \in \mathcal{C} \}
\end{equation}

\noindent
In the general case (studied in Section \ref{sec:constrainedCorrelatedEquilibrium}), coupled constraint sets (one per device) may be unrelated to each other, except that they are each a subset of strategies in an extension of the game $G$. This no longer holds in this section where the constraints induce, by definition, a new relation among the coupled constraint sets.
For each pair of correlation devices, the corresponding coupled constraint sets induce feasible probability distributions in $\mathcal{C}$.
The case without constraints (for any device $d$, $\mathcal{R}_d=\mathcal{S}_d$) is obtained by taking $\mathcal{C} = \Delta(\mathcal{A})$.
\subsection{Sufficiency of the canonical devices}

Let $\beta_i : {\mathcal{A}_i} \rightarrow {\mathcal{A}_i}$, $\bm{p}$ be a probability distribution in $\Delta(\mathcal{A})$ and
${\bm{z}_{\beta_i, \bm{p}}}$ be the probability distribution such that for any $\bm{a} \in \mathcal{A}$, 
\begin{equation}\label{eq:probability_distribution_for_deviation_with_canonical_device}
    \bm{z}_{\beta_i, \bm{p}}(\bm{a}) 
    = 
    \sum_{b_i \in \mathcal{A}_i} \bm{p}(b_i, \bm{a}_{-i}) \mathds{1}_{\beta_i(b_i) = a_i}
\end{equation}
interpreted as the probability distribution over action profiles such that any player $j$ plays her action component $b_j$ of the action profile $\bm{b}$ drawn from $\bm{p}$ except player $i$ unilaterally deviating and playing action $\beta_i(b_i)=a_i$ instead of $b_i$.

To simplify the notation, we have included only the action profile $\bm{a}$ as a variable in the  function $\bm{z}$,  but in general, the strategy $\beta_i$  and the probability distribution $\bm{p}$ are also variables of $\bm{z}$. Furthermore, in this section we denote $\beta_i \circ \alpha_i$ the usual composition of the functions $\beta_i:\mathcal{A}_i\rightarrow \mathcal{A}_i$ and $\alpha_i:\Omega\rightarrow \mathcal{A}_i$.

The following preliminary result shows the relation between the probability distribution $\bm{p}_{\bm{\alpha}}$ induced by a correlated strategy profile $\bm{\alpha} \in \mathcal{S}_{d}$, the deviation function $\beta_i : \mathcal{A}_i \rightarrow \mathcal{A}_i$ of player $i$ and $\bm{z}_{\beta_i, \bm{p}_{\bm{\alpha}}}$.
\begin{lemma}\label{lemma:distribution_induced_by_composition}
        Let $G$ be a finite non-cooperative game and $d$ a correlation device. 
	Furthermore, let $\bm{\alpha}\in \mathcal{S}_{d}$ and $\beta_i : {\mathcal{A}_i} \rightarrow {\mathcal{A}_i}$. 
        Then, for any $\bm{a}\in\mathcal{A}$,
        $\bm{p}_{(\beta_i\circ \alpha_i, \bm{\alpha}_{-i})}(\bm{a}) = {\bm{z}_{\beta_i, \bm{p}_{\bm{\alpha}}}}(\bm{a})$.
    \end{lemma}
    \begin{proof}[\textbf{Proof of Lemma \ref{lemma:distribution_induced_by_composition}}]
        Let $ \mathcal{P}_i = \{ P_{k}\}_{k \in \mathcal{A}_i}$ be a partition of $\Omega$ such that $P_{k} = \{ \omega \in \Omega : \alpha_i(\omega) = k \}$.
        For any $\bm{a}\in\mathcal{A}$, we have,
        \begin{align}
            \bm{p}_{(\beta_i\circ\alpha_i,\bm{\alpha}_{-i})}(\bm{a}) 
            &= 
            \sum_{\omega \in \Omega} \bm{q}(\omega) 
            \mathds{1}_{\{(\beta_i \circ \alpha_i, \bm{\alpha}_{-i})(\omega) = \bm{a}\}}
            \\
            &=
             \sum_{k \in \mathcal{A}_i} 
             \sum_{\omega \in {P}_{k}} \bm{q}(\omega) 
             \mathds{1}_{\{(\beta_i \circ \alpha_i, \bm{\alpha}_{-i})(\omega) = \bm{a}\}}
             \\
             &=
              \sum_{k \in \mathcal{A}_i} 
              \sum_{\omega \in {P}_{k}} \bm{q}(\omega) 
              \mathds{1}_{\{\beta_i \circ \alpha_i(\omega) = a_i\}}
              \mathds{1}_{\{\bm{\alpha}_{-i}(\omega) = \bm{a}_{-i}\}}
              \\
              &=
              \sum_{k \in \mathcal{A}_i} 
              \sum_{\omega \in {P}_{k}} \bm{q}(\omega) 
              \mathds{1}_{\{(\beta_i(\alpha_i(\omega)) = a_i\}}
              \mathds{1}_{\{\bm{\alpha}_{-i}(\omega) = \bm{a}_{-i}\}}
        \end{align}
        Since for any $\omega \in P_{k}$, $\alpha_i(\omega) = k$, we have
        \begin{align}
             \bm{p}_{(\beta_i\circ\alpha_i,\bm{\alpha}_{-i})}(\bm{a})
            &=
            \sum_{k \in \mathcal{A}_i} 
            \sum_{\omega \in {P}_{k}} \bm{q}(\omega) 
            \mathds{1}_{\{\beta_i(k) = a_i\}}
            \mathds{1}_{\{\bm{\alpha}_{-i}(\omega) = \bm{a}_{-i}\}}
            \\
            &=
            \sum_{k \in \mathcal{A}_i} 
            \mathds{1}_{\{\beta_i(k) = a_i\}}
            \sum_{\omega \in {P}_{k}} \bm{q}(\omega) 
            \mathds{1}_{\{\bm{\alpha}_{-i}(\omega) = \bm{a}_{-i}\}}
            \\
            &=
            \sum_{k \in \mathcal{A}_i} 
            \mathds{1}_{\{\beta_i(k) = a_i\}} \bm{p}_{\bm{\alpha}}(k, \bm{a}_{-i})
            \\
            &= 
            \bm{z}_{\beta_i, \bm{p}_{\bm{\alpha}}}
        \end{align}
    \end{proof}
\noindent The following theorem characterizes the set of constrained correlated equilibrium distributions.


\begin{theorem}\label{thm:canonicalEquivalence}
    Let $G$ be a finite non-cooperative game and $\mathcal{C}$ a set of feasible probability distributions. 
    The distribution $\bm{p}\in\Delta(\mathcal{A})$ is a constrained correlated equilibrium distribution 
    if and only if 
    for any player $i\in\mathcal{N}$,
    for any strategy $\beta_i : {\mathcal{A}_i} \rightarrow {\mathcal{A}_i}$, 
    if 
    $\bm{z}_{\beta_i,\bm{p}} \in \mathcal{C}$, 
    then
    \begin{align}
        \sum\limits_{\bm{a} \in \mathcal{A}} \bm{p}(\bm{a})
        \left[
            u_{i}(a_i, \bm{a}_{-i}) - u_{i}(\beta_i(a_i), \bm{a}_{-i}) 
        \right] 
        \geq 0 
        \label{eq:defnCanonical}
    \end{align}
\end{theorem}

\begin{proof}[\textbf{Proof of Theorem \ref{thm:canonicalEquivalence}}]
    First, we show by contradiction the contrapositive proposition of $(\Rightarrow\nobreak\hspace{0pt})$\footnote{\textbf{Proposition.} \emph{Let $\bm{p}\in\Delta(\mathcal{A})$.
            If 
            it exists a player $i \in \mathcal{N}$ 
            and a strategy $\beta_i : {\mathcal{A}_i} \rightarrow {\mathcal{A}_i}$ 
            such that,
            $\bm{z}_{\beta_i, \bm{p}} \in \mathcal{C}$ 
            and
            $
                \sum_{\bm{a} \in \mathcal{A}} \bm{p}(\bm{a})
                \left[
                    u_{i}(a_i, \bm{a}_{-i}) - u_{i}(\beta_i(a_i), \bm{a}_{-i}) 
                \right] 
                < 0 
            $
            then $\bm{p}$ is not a constrained correlated equilibrium distribution.}}.
    Let $\bm{p}$ be a constrained correlated equilibrium distribution. 
    Then, it exists a constrained correlated equilibrium 
    $(d,\mathcal{R}_d,\bm{\alpha}^*)$ such that $\bm{p} = \bm{p}_{\bm{\alpha}^*}$.
    Define
    \begin{equation}
        \mathcal{K}_{i,d}(\bm{\alpha}_{-i}^*) 
        = 
        \{ 
        	\alpha_i\in\mathcal{S}_{i,d} 
        	\mid  
        	\bm{p}_{(\alpha_i,\bm{\alpha}_{-i}^*)}\in\mathcal{C} 
        \}
    \end{equation} 
    
    \noindent Furthermore, let 
    $\mathcal{H}_{i}(\bm{p}) = \{\beta^\prime_i : {\mathcal{A}_i} \rightarrow {\mathcal{A}_i} \mid \bm{z}_{\beta^\prime_i, \bm{p}}\in\mathcal{C}\}$ 
    and assume $i\in\mathcal{N}$, $\beta_i\in \mathcal{H}_{i}(\bm{p})$ such that
	$\bm{z}_{\beta_i,\bm{p}} \in \mathcal{C}$ and
    \begin{equation}
        \sum\limits_{\bm{a} \in \mathcal{A}} \bm{p}(\bm{a})
        \left[
            u_{i}(a_i, \bm{a}_{-i}) - u_{i}(\beta_i(a_i), \bm{a}_{-i}) 
        \right] 
        < 0 
    \end{equation} 

    \noindent From Lemma \ref{lemma:distribution_induced_by_composition}, we have
    $
        \bm{p}_{(\beta_i\circ \alpha_i^*,\bm{\alpha^*}_{-i})} = {\bm{z}_{\beta_i, \bm{p}}}
    $. 
    Since $\bm{z}_{\beta_i, \bm{p}} \in \mathcal{C}$, by definition of $\mathcal{K}_{i,d}$, we have $\beta_i\circ \alpha_i^*\in \mathcal{K}_{i,d}(\bm{\alpha}^*_{-i})$.
    \noindent Furthermore, 
    \begin{align}
        \sum_{\omega\in\Omega} \bm{q}(\omega) 
        \left[
            u_i(\alpha_i^*(\omega),\bm{\alpha}_{-i}^*(\omega)) 
            -
            u_i(\beta_i\circ \alpha_i^*(\omega),\bm{\alpha}_{-i}^*(\omega))
        \right]
        &= 
        \sum_{\bm{a}\in\mathcal{A}}\bm{p}_{\bm{\alpha}^*}(\bm{a})
        \left[
            u_i(a_i,\bm{a}_{-i}) 
            -
            u_i(\beta_i(a_i),\bm{a}_{-i})
        \right]\\
        &= 
        \sum_{\bm{a}\in\mathcal{A}}\bm{p}(\bm{a})
        \left[
            u_i(a_i,\bm{a}_{-i}) 
            -
            u_i(\beta_i(a_i),\bm{a}_{-i})
        \right]
        <0 \label{eq:contradiction}
    \end{align}
    where the last equality follows from the definition of $\bm{p}$ and the inequality follows from the assumption on $\beta_i$.
    Thus, $\bm{\alpha}^*$ is not a constrained correlated equilibrium implying that $\bm{p}$ is not a constrained correlated equilibrium. A contradiction.\\
    
    \noindent ($\Leftarrow$) Let $\bm{p}\in\Delta(\mathcal{A})$  such that $\bm{p}\in \mathcal{C}$ and 
    for any $i \in \mathcal{N}$, 
    for any $\beta_i : {\mathcal{A}_i} \rightarrow {\mathcal{A}_i}$, 
    $\bm{z}_{\beta_i, \bm{p}} \in \mathcal{C}$ implies
    \begin{align}
        \sum\limits_{\bm{a} \in \mathcal{A}} \bm{p}(\bm{a}) 
        \left[
            u_{i}(a_i, \bm{a}_{-i}) - u_{i}(\beta_i(a_i), \bm{a}_{-i}) 
        \right] \geq 0 
        \label{DP}
    \end{align}

To prove that $\bm{p}$ is a constrained correlated equilibrium distribution, we must show that there exists a constrained correlated equilibrium $(d,\mathcal{R}_d,\bm{\alpha}^*)$ such that $\bm{p} = \bm{p}_{\bm{\alpha}^*}$. 

Let $d=(\mathcal{A},(\mathcal{P}_i)_{i\in\mathcal{N}},\bm{p})$ 
    be the canonical correlation device such that
    $\mathcal{P}_i = \{  P_{k} \}_{k \in \mathcal{A}_i}$ is a partition of $\Omega$ where 
    $P_{k} = \{ \omega \in \Omega : \alpha_i(\omega) = k \}$
    and 
    $\alpha_i^*:\mathcal{A} \rightarrow  \mathcal{A}_i$ 
    such that for any $\bm{a\in\mathcal{A}}$,
    $\alpha_i^*(\bm{a})=a_i$. 

We show that for any  $\alpha_i^\prime\in\mathcal{K}_{i,d}(\bm{\alpha}^*_{-i})$, 
    there exists 
    $\beta_i : {\mathcal{A}_i} \rightarrow {\mathcal{A}_i}$ such that $\alpha_i^\prime = \beta_i\circ\alpha_i^*$ and $\bm{z}_{\beta_i, \bm{p}} \in \mathcal{C}$. First, we show the existence of such $\beta_i$.

\noindent Let us consider (without loss of generality) the strategy $\alpha_i^\prime$ defined as 
$$
    \alpha_i^\prime(a_i, \bm{a}_{-i}) = b_{a_i},\; a_i\in \mathcal{A}_i,\;  \bm{a}_{-i} \in \mathcal{A}_{-i},
$$ 
where $b_{a_i}$ is an action in $\mathcal{A}_{i}$. 
This  means   that player $i$ plays action $b_{a_i}$ in $\mathcal{A}_i$ if action $a_i$ has been recommended by the canonical device $d_c$.  Let $\beta_i : {\mathcal{A}_i} \rightarrow {\mathcal{A}_i}$ such that $\beta_i(a_i) = b_{a_i}$, $a_i\in\mathcal{A}_i$. Thus we have for any $\bm{a} \in P_{a_i}$: 
    \begin{equation}
        \beta_i\circ\alpha_i^*(\bm{a}) 
        = \beta_i(\alpha_i^*(\bm{a})) 
        = \beta_i(a_i)
        = b_{a_i}
        = \alpha^\prime_i(\bm{a})
 \end{equation}
Thus, $\beta_i\circ\alpha_i^* = \alpha_i^\prime$. 
  \noindent 
  Second, we show that $\bm{z}_{\beta_i, \bm{p}} \in \mathcal{C}$.
    From lemma \ref{lemma:distribution_induced_by_composition}, we have, 
    $
    \bm{p}_{(\beta_i\circ\alpha_i^*,\bm{\alpha}_{-i}^*)} = \bm{z}_{\beta_i, \bm{p}}
    $
    Furthermore 
    $\bm{p}_{(\beta_i\circ\alpha_i^*,\bm{\alpha}_{-i}^*)} = \bm{p}_{(\alpha_i^\prime,\bm{\alpha}_{-i}^*)}$ 
    and $\bm{p}_{(\alpha_i^\prime,\bm{\alpha}_{-i}^*)}\in \mathcal{C}$ 
    (since $\alpha_i^\prime\in\mathcal{K}_{i,d}(\bm{\alpha}^*_{-i})$ implies $\bm{p}_{(\alpha_i^\prime,\bm{\alpha}_{-i}^*)}\in\mathcal{C}$). 
    Then, $\bm{z}_{\beta_i, \bm{p}}\in\mathcal{C}$. 

    \noindent 
    Third, we show that $\bm{\alpha}^*$ is a constrained correlated equilibrium.
    We have, 
    \begin{align}  
        \sum\limits_{\omega \in \Omega} \bm{q}(\omega) 
        \left[
            u_{i}(\bm{\alpha}^*(\omega)) - u_{i}(\alpha_i^{\prime}(\omega), \bm{\alpha}^*_{-i}(\omega)) 
        \right]
        & =     
        \sum\limits_{k \in \mathcal{A}_i} \sum\limits_{\omega \in P_{k}} \bm{p}(\omega) 
        \left[
            u_{i}(\bm{\alpha}^*(\omega)) - u_{i}(\alpha_i^\prime(\omega), \bm{\alpha}^*_{-i}(\omega)) 
        \right]
        \\
        & =
        \sum\limits_{k \in \mathcal{A}_i} \sum\limits_{\omega \in P_{k}} \bm{p}(\omega) 
        \left[
            u_{i}(\bm{\alpha}^*(\omega)) - u_{i}(\beta_i\circ \alpha^*_i(\omega), \bm{\alpha}^*_{-i}(\omega)) 
        \right]
        \\ 
        & =
        \sum\limits_{k \in \mathcal{A}_i} \sum\limits_{\omega \in P_{k}} \bm{p}(\omega) 
        \left[
            u_{i}(\bm{\alpha}^*(\omega)) - u_{i}(\beta_i(k), \bm{\alpha}^*_{-i}(\omega)) 
        \right]
        \\         
        & = \sum\limits_{\bm{k} \in \mathcal{A}} 
        \bm{p}(\bm{k}) 
        \left[
            u_{i}(k_i, \bm{k}_{-i}) - u_{i}(\beta_i(k_i), \bm{k}_{-i}) 
        \right]
        \geq 0 \label{eqDP1}
     \end{align}    
    where the last inequality follows by assumption. 
    Thus $(d,\mathcal{R}_d,\bm{\alpha}^*)$ is a constrained correlated equilibrium, implying that $\bm{p} = \bm{p}_{\bm{\alpha}^*}$ is a constrained correlated equilibrium distribution. 
\end{proof}

The latter theorem implies that, if the coupled constraint sets are induced by a feasible set of probability distributions, it is sufficient to consider canonical devices to characterize the set of constrained correlated equilibrium distributions. 
In the case without constraints, obtained by taking $\mathcal{C} = \Delta(\mathcal{A})$, the equilibrium conditions from Theorem \ref{thm:canonicalEquivalence} reduce to \eqref{eq:correlatedEquilibriumDistributions}.

The following result is a corollary result of Proposition \ref{prop:robustnessToRestrictions} showing sufficient conditions for a constrained correlated equilibrium when the constraints are generated by feasible sets of probability distributions.
\begin{corollary}\label{cor:SmallerConstraints}
    Let $\mathcal{C}$ and $\mathcal{C}^\prime$ be feasible sets of probability distributions such that $\mathcal{C} \subseteq \mathcal{C}^\prime$. Let $d$ be a correlation device, $\mathcal{R}_d$ and $\mathcal{R}^\prime_d$ be the feasible sets of probability distributions generated by $\mathcal{C}$ and $\mathcal{C}^\prime$ respectively.
    If $(d, \mathcal{R}^\prime_d, \bm{\alpha}^*)$ is a constrained correlated equilibrium  
    and $\bm{p}_{\bm{\alpha}^*} \in \mathcal{C}$, then $(d, \mathcal{R}_d, \bm{\alpha}^*)$ is a constrained correlated equilibrium.
\end{corollary}
\begin{proof}[\textbf{Proof of Corollary~\ref{cor:SmallerConstraints}}]
    From the definition of $\mathcal{R}_{d}$ and $\mathcal{R}^\prime_{d}$ and with the assumption $\mathcal{C} \subseteq \mathcal{C}^\prime$, we have $\mathcal{R}_{d} \subseteq \mathcal{R}^\prime_{d}$.
    According to Proposition \ref{prop:robustnessToRestrictions}, the proof is immediate.
\end{proof}
\subsection{Existence of Constrained Correlated Equilibrium }


The existence of correlated equilibria has been shown in \cite{aumann1987correlated} and \cite{hart1989existence}. 
However, in case of constraints, a general existence result does not hold even if restricting the scope to the case of coupled constraint sets induced by a feasible set of probability distributions as considered in this section.
The following example has no constrained correlated equilibrium.\\


\noindent \textbf{Example.} 
Let $G$ be the game shown in Figure \ref{tab:2x2game} with one mixed Nash equilibrium $((\sfrac{5}{6} \cdot U, \sfrac{1}{6} \cdot D), (\sfrac{1}{2} \cdot L, \sfrac{1}{2} \cdot R))$ and 
$
    \mathcal{C} = \{ \bm{p} \in \Delta(\mathcal{A}) \mid \bm{p}(U,L) = 1 \text{ or } \bm{p}(U,R) = 1 \text{ or } \bm{p}(D,L) = 1 \text{ or } \bm{p}(D,R) = 1\}
$ 
be a feasible set of probability distributions such that the players must play a correlated strategy profile inducing a pure action profile in $G$. For any device $d$, the coupled constraints set $\mathcal{R}_d$ generated by $\mathcal{C}$ implies that, the feasible strategies of the game $G_d$  are the set of strategies where each player uses the same action  regardless of the outcome $\omega \in \Omega$. We then recover the game $G$ and since it does not admit a pure Nash equilibrium, the same applies to the game $G_d$.  We conclude  that there is no constrained correlated equilibrium.

\begin{figure}[H]
    \normalsize
    \centering
        \scalebox{1.2}{
        \begin{tabular}{c|c|c}
                  & \multicolumn{1}{c|}{{$L$}} & {$R$}  \\ \hline
            {$\,\,\, U \,\,\,$} & \multicolumn{1}{c|}{$(2, 2)$} & $(1, 1)$  \\ \hline
            {$\,\,\, D \,\,\,$} & \multicolumn{1}{c|}{$(3, 0)$} & $(0, 5)$
        \end{tabular}}
    \caption{Two-player game in matrix form.}
    \label{tab:2x2game}
\end{figure}
The non-existence of a constrained correlated equilibrium can also be observed in any game $G$ without a pure Nash equilibrium and such that the coupled constraints require that players play one of the  action profiles with probability one, \ie $
    \mathcal{C} = \{ \bm{p} \in \Delta(\mathcal{A}) \mid \forall \bm{a}\in \mathcal{A}, \; \bm{p}({\bm{a}})\in \{0,1\} \}
$.


This previous example shows that further assumptions are required for the existence of constrained correlated equilibria.
Before giving sufficient conditions of existence (Theorem~\ref{thm:existence} and Proposition \ref{prop:4}), we show the following technical  result (used in the proof of Theorem \ref{thm:existence}).

\begin{lemma}\label{lem:expected_utility_after_deviation}
        If $\bm{p} \in \Delta(\mathcal{A})$ and $\beta_i : {\mathcal{A}_i} \rightarrow {\mathcal{A}_i}$ then, $$u_i(\bm{z}_{\beta_i, \bm{p}}) = \sum\limits_{\bm{a} \in \mathcal{A}} \bm{p}(\bm{a}) u_i(\beta_i(a_i), \bm{a}_{-i})$$
    \end{lemma}

\begin{proof}[\textbf{Proof of Lemma \ref{lem:expected_utility_after_deviation}}]
    Let $\bm{\alpha}^*$ be a correlated strategy profile. By definition, for any player $i$, the expected utility is such that,
    \begin{align}
        \sum_{\omega \in \Omega} 
            \bm{q}(\omega)
            u_i(\alpha^*_i(\omega), \bm{\alpha}_{-i}^*(\omega))
        = & \sum_{\bm{a} \in \mathcal{A}} \bm{p}_{(\alpha^*_i, \bm{\alpha}_{-i}^*)}(\bm{a}) u_i(\bm{a})
    \end{align}
    Let $\beta_i : {\mathcal{A}_i} \rightarrow {\mathcal{A}_i}$. For the correlated strategy profile $(\beta_i \circ \alpha_i^*, \bm{\alpha}_{-i}^*)$, we similarly have,
    \begin{align}
        \sum_{\omega \in \Omega} \bm{q}(\omega) u_i(\beta_i \circ \alpha^*_i(\omega), \bm{\alpha}_{-i}^*(\omega))
        = & \sum_{\bm{a} \in \mathcal{A}} \bm{p}_{(\beta_i \circ \alpha^*_i, \bm{\alpha}_{-i}^*)}(\bm{a}) u_i(\bm{a})
    \end{align}
    On the other hand, from \eqref{eq:contradiction}, we also have
    \begin{align}
        \sum_{\omega \in \Omega} \bm{q}(\omega) u_i(\beta_i \circ \alpha^*_i(\omega), \bm{\alpha}_{-i}^*(\omega))
        = & \sum_{\bm{a} \in \mathcal{A}} \bm{p}_{(\alpha^*_i, \bm{\alpha}_{-i}^*)}(\bm{a}) u_i(\beta_i(a_i), \bm{a}_{-i})
    \end{align}
    Thus, according to Lemma \ref{lemma:distribution_induced_by_composition},
    \begin{align}
        \sum_{\bm{a} \in \mathcal{A}} \bm{p}_{(\beta_i \circ \alpha^*_i, \bm{\alpha}_{-i}^*)}(\bm{a}) u_i(\bm{a}) = \sum_{\bm{a} \in \mathcal{A}} \bm{z}_{\beta_i, \bm{p}_{\bm{\alpha}^*}}(\bm{a}) u_i(\bm{a})
    \end{align}
    Therefore,
    \begin{align}
        u_i(\bm{z}_{\beta_i,\bm{p}_{\bm{\alpha}^*}}) =
        \sum_{\bm{a} \in \mathcal{A}} \bm{p}_{\bm{\alpha}^*}(\bm{a}) u_i(\beta_i(a_i), \bm{a}_{-i})
    \end{align}
\end{proof}
The following theorem shows that the convexity and compactness of the feasible set of probability distributions $\mathcal{C}$ implies the existence of constrained correlated equilibria.

\begin{theorem}\label{thm:existence}
    Let $G$ be a finite non-cooperative game and $\mathcal{C}$ a feasible set of probability distributions. 
    If $\mathcal{C}$ is non-empty, compact and convex, then a constrained correlated equilibrium of $G$ exists. 
\end{theorem}
\begin{proof}[\textbf{Proof of Theorem \ref{thm:existence}}]
    The proof is similar to the proof of Nash's theorem \cite{nash_non-cooperative_1951} using Brouwer's fixed-point theorem.

    \noindent Let $\mathcal{C} \subseteq \Delta(\mathcal{A})$ be a non-empty compact and convex set and $\bm{p} \in \mathcal{C}$.
    For any $i \in \mathcal{N}$ 
    and 
    any 
    $\beta_i \in \mathcal{H}_i(\bm{p}) = \{ f_i : {\mathcal{A}_i} \rightarrow {\mathcal{A}_i} \mid  \bm{z}_{f_i,\bm{p}} \in \mathcal{C}\}$, define $\varphi_{i, \beta_i}:\mathcal{C}\rightarrow \mathbb{R}$ such that,
        \begin{align}
            \varphi_{i, \beta_i}(\bm{p})
            =
            \max \{ 0, u_i({\bm{z}_{\beta_i, \bm{p}}})-u_i(\bm{p}) \} 
            \label{eq:varPhi}
        \end{align}
    Furthermore, define $g: \mathcal{C} \rightarrow \mathcal{C}$ such that for any $\bm{p}\in\mathcal{C}$, 
    \begin{multline}
        g(\bm{p}) =
            \left(
                1 - 
                \frac{
                        \sum\limits_{i\in\mathcal{N}}
                        \sum\limits_{\beta_i \in \mathcal{H}_i(\bm{p})} \varphi_{i,\beta_i}(\bm{p})
                    }
                    {
                        1+\sum\limits_{i\in\mathcal{N}}\sum\limits_{\alpha_i \in \mathcal{H}_i(\bm{p})} \varphi_{i,\alpha_i}(\bm{p})
                    }
            \right)
            \times
            \bm{p}
            +
            \sum_{i \in \mathcal{N}}
            \sum_{\beta_i\in \mathcal{H}_i(\bm{p})}
            \frac{\varphi_{i, \beta_i}(\bm{p})}{1+ \sum\limits_{i\in\mathcal{N}}\sum\limits_{\alpha_i\in \mathcal{H}_i(\bm{p})}\varphi_{i,\alpha_i}(\bm{p})}
            \times
            {\bm{z}_{\beta_i, \bm{p}}}
            \label{eq:functionG}
    \end{multline}
    This function maps any probability distribution $\bm{p}\in\mathcal{C}$ to $g(\bm{p})$ defined as a convex combination of $\bm{p}$ and the probability distributions in the set $\{z_{\beta_i,\bm{p}}\}_{\beta_i\in\mathcal{H}_i(\bm{p})}$, each belonging to $\mathcal{C}$ by definition of $\mathcal{H}_i(\bm{p})$.
    Then, by convexity of $\mathcal{C}$, we have $g(\bm{p})\in\mathcal{C}$. 
    Furthermore, $g$ is continuous 
    as a composition of continuous functions (the expected utility function $u_i$ is continuous in $\bm{p}$ and Lemma \ref{lem:expected_utility_after_deviation} in the Appendix shows the continuity of the function $\bm{z}_{\beta_i,\bm{p}}$). 
    Since $\mathcal{C}$ is compact and convex, by Brouwer's fixed point theorem\footnote{
    \textbf{Brouwer's fixed point theorem \cite{florenzano2003general}.}
        \emph{Let $X$ be a nonempty compact convex subset of $\mathbb{R}^{\ell}$ and $f: X \rightarrow X$ a continuous (single-valued) mapping. Then there exists an $\bar{x} \in X$ such that $f(\bar{x})=\bar{x}$.}
    },
    $g$ has a fixed point.

    \noindent
    The following shows that a probability distribution is a constrained correlated equilibrium distribution if and only if it is a fixed point of $g$.

    \noindent
    ($\Rightarrow$)
    Assume that $\bm{p}^*$ is a constrained correlated equilibrium distribution. 
  
    \noindent 
    Theorem \ref{thm:canonicalEquivalence} implies 
    that for any $i \in \mathcal{N}$ and any $\beta_i : {\mathcal{A}_i} \rightarrow {\mathcal{A}_i}$ such that $\bm{z}_{\beta_i, \bm{p}^*} \in \mathcal{C}$, we have
    \begin{align}
        \sum\limits_{\bm{a} \in \mathcal{A}} \bm{p}^*(\bm{a}) 
        \left[
            u_{i}(a_i, \bm{a}_{-i}) - u_{i}(\beta_i(a_i), \bm{a}_{-i}) 
        \right] \geq 0
    \end{align}
    Furthermore, from Lemma \ref{lem:expected_utility_after_deviation},
    \begin{equation}
        u_i(\bm{z}_{\beta_i, \bm{p}^*}) 
        = 
        \sum_{\bm{a} \in \mathcal{A}}\bm{p}^*(\bm{a})u_i(\beta_i(a_i), \bm{a}_{-i})
    \end{equation}
    
\noindent Then,
    \begin{equation}
        \sum\limits_{\bm{a} \in \mathcal{A}}
            \bm{p}^*(\bm{a}) u_{i}(\bm{a}) 
        -
        \sum\limits_{\bm{a} \in \mathcal{A}}
            \bm{z}_{\beta_i, \bm{p}^*}(\bm{a}) u_{i}(\bm{a}) 
        \geq 0 
    \end{equation} 
    Thus, for any $\beta_i\in\mathcal{H}_i(\bm{p}^*)$, $u_i(\bm{p}^*)\geq u_i(\bm{z}_{\beta_i,\bm{p}^*})$, implying $\varphi_{i,\beta_i}(\bm{p}^*)=0$.
    Then $g(\bm{p}^*) = \bm{p}^*$, \ie $\bm{p}^*$ is a fixed point of $g$.
    \\ 

    \noindent
    ($\Leftarrow$) Conversely, assume that $\bm{p}^*$ is a fixed point of $g$. 
    \\
    
    \noindent
    \textbf{Case 1.} For any $i \in \mathcal{N}$ and any $ \beta_i \in \mathcal{H}_i(\bm{p}^*)$, $\varphi_{i,\beta_i}(\bm{p}^*) = 0$. 
    
    \noindent Then, 
    by definition of $\varphi_{i,\beta_i}$ and $\mathcal{H}_i(\bm{p}^*)$, for any $i \in \mathcal{N}$ and any $ \beta_i : {\mathcal{A}_i} \rightarrow {\mathcal{A}_i}$ s.t. $\bm{z}_{\beta_i,\bm{p}^*}\in\mathcal{C}$, we have $u_i(\bm{z}_{\beta_i, \bm{p}^*}) - u_i(\bm{p}^*) \leq 0$, implying
    \begin{equation}
        \sum_{\bm{a} \in \mathcal{A}}
        \bm{p}^*(\bm{a})[u_i(\bm{a}) - u_i(\beta_i(a_i), \bm{a}_{-i})] \geq 0    
    \end{equation}
    Thus, from Theorem \ref{thm:canonicalEquivalence}, $\bm{p}^*$ is a constrained correlated equilibrium distribution, implying that it exists a constrained correlated equilibrium. 
    \\

    \noindent
    \textbf{Case 2.} There exists $\beta_i \in \mathcal{H}_i(\bm{p}^*)$ such that $\varphi_{i, \beta_i}(\bm{p}^*) >0$. 

    \noindent
    Let
    $
    \mathcal{H}_i^+(\bm{p}^*) 
    = 
    \{
        f_i\in\mathcal{H}_i(\bm{p}^*) \mid u_i(\bm{z}_{f_i,\bm{p}^*})> u_i(\bm{p}^*)
    \}
    $ be the set of profitable feasible deviations for player $i$.  
    For each $\beta_i\in\mathcal{H}_i^+(\bm{p}^*)$, we have 
    $\varphi_{i, \beta_i}(\bm{p}^*) = u_i(\bm{z}_{\beta_i,\bm{p}^*})-u_i(\bm{p}^*) > 0$ 
    and for each 
    $\beta_i\in \mathcal{H}_i(\bm{p}^*)\backslash \mathcal{H}_i^+(\bm{p}^*)$, 
    we have 
    $\varphi_{i, \beta_i}(\bm{p}^*) = 0$.
    Then,
    \begin{multline}
        \bm{p}^*  =
            \left(
                1-
                \frac{
                    \sum\limits_{i\in\mathcal{N}}
                    \sum\limits_{\beta_i\in\mathcal{H}_i^+(\bm{p}^*)} 
                    \varphi_{i,\beta_i}(\bm{p}^*)
                }
                {
                    1+
                    \sum\limits_{i\in\mathcal{N}}
                    \sum\limits_{\alpha_i\in\mathcal{H}_i^+(\bm{p}^*)} 
                    \varphi_{i,\alpha_i}(\bm{p}^*)
                    }
            \right) 
            \times \bm{p}^* +
            \sum_{i \in \mathcal{N}}\sum_{\beta_i\in\mathcal{H}_i^+(\bm{p}^*)} \frac{\varphi_{i, \beta_i}(\bm{p})}{1+ \sum\limits_{i\in\mathcal{N}}\sum\limits_{\alpha_i\in\mathcal{H}_i^+(\bm{p}^*)}\varphi_{i,\alpha_i}(\bm{p}^*)} \times {\bm{z}_{\beta_i, \bm{p}^*}}
    \end{multline}
    Implying,
    \begin{equation}
        \sum\limits_{i\in\mathcal{N}}
        \sum\limits_{\alpha_i\in\mathcal{H}_i^+(\bm{p}^*)}
        \varphi_{i,\alpha_i}(\bm{p}^*) \times \bm{p}^* 
        =
        \sum_{i \in \mathcal{N}}\sum_{\beta_i\in\mathcal{H}_i^+(\bm{p}^*)}
        \varphi_{i, \beta_i}(\bm{p}^*)
        \times
        {\bm{z}_{\beta_i, \bm{p}^*}}
        \label{eq:fixedPoint}
    \end{equation}
    Dividing by 
    $       
        \sum\limits_{i\in\mathcal{N}}\sum\limits_{\alpha_i\in\mathcal{H}_i^+(\bm{p}^*)}\varphi_{i,\alpha_i}(\bm{p}^*)
        >0
    $
    on both sides, we obtain
    \begin{align}
        \bm{p}^* & =
        \sum_{i \in \mathcal{N}}
        \sum_{\beta_i\in\mathcal{H}_i^+(\bm{p}^*)}
        \frac{\varphi_{i, \beta_i}(\bm{p}^*)}{\sum\limits_{i\in\mathcal{N}}\sum\limits_{\alpha_i\in\mathcal{H}_i^+(\bm{p}^*)} \varphi_{i,\alpha_i}(\bm{p}^*)} \times {\bm{z}_{\beta_i, \bm{p}^*}}
    \end{align}
 
    \noindent
    Then, by linearity of $u_i$, we have
    \begin{equation}
        u_i(\bm{p}^*) 
        =
        \sum_{i \in \mathcal{N}}
        \sum_{\beta_i\in\mathcal{H}_i^+(\bm{p}^*)}
        \frac{\varphi_{i, \beta_i}(\bm{p}^*)}{\sum\limits_{i\in\mathcal{N}}\sum\limits_{\alpha_i\in\mathcal{H}_i^+(\bm{p}^*)} \varphi_{i,\alpha_i}(\bm{p}^*)}
        u_i({\bm{z}_{\beta_i, \bm{p}^*}})
        \label{eq:payoffFunction}
    \end{equation}
    Thus, 
    \begin{align}
        u_i(\bm{p}^*) - u_i(\bm{p}^*) 
        &= 
        \left(
        \sum_{i \in \mathcal{N}}
        \sum_{\beta_i\in\mathcal{H}_i^+(\bm{p}^*)}
        \frac{\varphi_{i, \beta_i}(\bm{p}^*)}{\sum\limits_{i\in\mathcal{N}}\sum\limits_{\alpha_i\in\mathcal{H}_i^+(\bm{p}^*)} \varphi_{i,\alpha_i}(\bm{p}^*)}
        u_i({\bm{z}_{\beta_i, \bm{p}^*}}) 
        \right)
        - 
        u_i(\bm{p}^*) \nonumber \\
        &= 
        \sum_{i \in \mathcal{N}}
        \sum_{\beta_i\in\mathcal{H}_i^+(\bm{p}^*)}
        \frac{\varphi_{i, \beta_i}(\bm{p}^*)}{\sum\limits_{i\in\mathcal{N}}
        \sum\limits_{\alpha_i\in\mathcal{H}_i^+(\bm{p}^*)}
        \varphi_{i,\alpha_i}(\bm{p}^*)}
        (u_i({\bm{z}_{\beta_i, \bm{p}^*}})-u_i(\bm{p}^*))  \nonumber    
    \end{align}
    Since, for each $\beta_i\in \mathcal{H}_i^+(\bm{p}^*)$, $\varphi_{i,\alpha_i}(\bm{p}^*)>0$ and $u_i({\bm{z}_{\beta_i, \bm{p}^*}})-u_i(\bm{p}^*)>0$, we have $u_i(\bm{p}^*) - u_i(\bm{p}^*)>0$, a contradiction. 
    \\ 
    

    \noindent
    Thus, if $\bm{p}^*$ is a fixed point of $g$, then for any $i \in \mathcal{N}$ and any $ \alpha_i \in \mathcal{H}_i(\bm{p}^*)$, $\varphi_{i,\alpha_i}(\bm{p}^*) = 0$, implying that $\bm{p}^*$ is a constrained correlated equilibrium distribution and the existence of a constrained correlated equilibrium.
\end{proof}

As already observed, the case without constraints is obtained by taking $\mathcal{C} = \Delta(\mathcal{A})$ which satisfies the conditions of Theorem \ref{thm:existence}, showing the existence of correlated equilibria.\\ 


\noindent 
The following result shows that any correlated equilibrium distribution in $\mathcal{C}$ is a constrained correlated equilibrium distribution.



\begin{proposition}\label{prop:4}
    Let $G$ be a finite non-cooperative game and $\mathcal{C}$ a feasible set of probability distributions. 
    If $\bm{p}^*\in\mathcal{C} \cap \mathcal{D}$ then $\bm{p}^*$ is a constrained correlated equilibrium distribution.
\end{proposition}
\begin{proof}[\textbf{Proof of Proposition \ref{prop:4}}]
The proof is immediate from Proposition \ref{prop:CEinCCE}. Indeed, since $\bm{p}^*\in\mathcal{D}$, there exists a device $d$ and a correlated equilibrium $\bm{\alpha}^*$ such that $\bm{p}_{\bm{\alpha}^*}=\bm{p}^*$. In addition, $\bm{\alpha}^*$ is a feasible strategy in $\mathcal{R}_d =\{\bm{\alpha} \in \mathcal{S}_d \mid \bm{p}_{\bm{\alpha}} \in \mathcal{C} \}$ and $\mathcal{R}_d \subseteq \mathcal{S}_d$, then according to Proposition \ref{prop:CEinCCE}, $(d,\mathcal{R}_d, \bm{\bm{\alpha}}^*)$  is also a constrained  correlated equilibrium and thus $\bm{p}^*$ is a constrained correlated equilibrium distribution.
\end{proof}
This result also shows that a non-empty intersection between $\mathcal{C}$ and $\mathcal{D}$ implies the existence of a constrained correlated equilibrium.


\subsection{Correlation in the mixed extension}

In this section, we consider the constrained correlated equilibrium distributions of the mixed extension of $G$, denoted $\Delta G$ (see Section \ref{sec:background}), and study whether or not new equilibrium distributions are obtained.
This problem has already been considered in \cite{aumann1987correlated} and \cite{fudenberg1991game} in the case without constraints, showing that the correlated equilibrium concept ``does not require explicit randomization on the part of the players" \cite{aumann1987correlated}.
In fact, any correlated equilibrium distribution of the mixed extension of $G$ can be obtained as a correlated equilibrium distribution of $G$ by a relevant choice of correlation device and correlated strategy profile. 
We show that a similar results holds when assuming a feasible set of probability distributions $\mathcal{C}$.


Let $d$ be a correlation device and $\Delta G_d$ be the extension of $\Delta G$ by $d$. 
In $\Delta G_d$, a strategy $\gamma_i$ for player is a $\mathcal{P}_i$-measurable function\footnote{Assuming $\Delta(\mathcal{A}_i)$ equipped with the Borel $\sigma$-algebra induced by the standard subspace topology on $\Delta(\mathcal{A}_i)$.}
mapping each outcome $\omega\in\Omega$ to a mixed strategy in $\Delta(\mathcal{A}_i)$. 
For any $\omega\in\Omega$, $\gamma_{i}(\omega)$ is a probability distribution in $\Delta(\mathcal{A}_i)$ and for any $a_i\in\mathcal{A}_i$, $\gamma_{i}(\omega)(a_i)$ is the probability of player $i$ choosing action $a_i\in \mathcal{A}_i$.
The set of strategies of player $i$ in $\Delta G_d$ is
\begin{equation}
    \tilde{\mathcal{S}}_{i, d} = \{ \gamma_i : \Omega \rightarrow \Delta(\mathcal{A}_i)
    \mid 
    \gamma_i
    \text{  is }
    \mathcal{P}_i\text{-measurable} \} \label{eq:mixedCorrelatedStrategy}
\end{equation}
and for any correlated strategy profile 
$\bm{\gamma} \in \tilde{\mathcal{S}}_d = \times_{i\in\mathcal{N}}\tilde{\mathcal{S}}_{i, d}$, the utility function of player 
$i$ is 
$
    \tilde{u}_i:\tilde{\mathcal{S}}_d \rightarrow \mathbb{R}
$ 
such that,
\begin{align}
     \tilde{u}_i(\gamma_i, \bm{\gamma}_{-i})
     &=
     \sum_{\omega \in \Omega}
     \bm{q}(\omega)
     u_i(\gamma_i(\omega),\bm{\alpha}_{-i}(\omega))
     \\
     &=
     \sum_{\omega \in \Omega}
     \bm{q}(\omega)
     \sum_{\bm{a} \in \mathcal{A}}
     \Big(
        \prod_{j \in \mathcal{N}}
            \gamma_{j}(\omega)(a_j)
    \Big)
    u_i(a_i,\bm{a}_{-i})
    \\
    &= 
    \sum_{\bm{a} \in \mathcal{A}}
    \sum_{\omega \in \Omega}     
     \bm{q}(\omega)
     \Big(
        \prod_{j \in \mathcal{N}}
            \gamma_{j}(\omega)(a_j)
    \Big)
    u_i(a_i,\bm{a}_{-i})    
     \label{mix-u}
\end{align}
where the second equality is obtained by definition of the utility function in $\Delta G$.

\noindent
 Let $\bm{p}_{\bm{\gamma}}\in\Delta(\mathcal{A})$ be the probability distribution induced by $\bm{\gamma}$ such that, for any $\bm{a}\in\mathcal{A}$, 
\begin{equation}
    \label{eq:probaDistributionInduced}
    \bm{p}_{\bm{\gamma}}(\bm{a})=\sum_{\omega \in \Omega} \bm{q}(\omega) 
    \left(
        \prod_{j \in \mathcal{N}} \gamma_j(\omega)(a_j)
    \right)
\end{equation}
Remark that the probability distribution $\bm{p}_{\bm{\gamma}}$ is not the probability distribution of $\bm{\gamma}$ (by definition) and any correlated strategy profile in $\Delta G_d$ induces a distribution as defined in \eqref{eq:probaDistributionInduced}.

\noindent
For any coupled constraint set $\tilde{\mathcal{R}}_d \subseteq \tilde{\mathcal{S}}_d$, 
the triplet $(d,\tilde{\mathcal{R}}_d,\bm{\gamma}^*)$
is a constrained correlated equilibrium of $\Delta G$ if $\bm{\gamma}^*\in \tilde{\mathcal{R}}_d$ 
and
for any $i \in \mathcal{N}$, 
for any $\gamma^\prime \in \tilde{\mathcal{K}}_{i,d}(\bm{\gamma}^*_{-i})$,
\begin{equation}
  \tilde{u}_i(\gamma^*_{i}, \bm{\gamma}^*_{-i}) \geq \tilde{u}_i(\gamma_i^\prime, \bm{\gamma}^*_{-i})
\end{equation}
where $\tilde{\mathcal{K}}_{i,d} : \tilde{\mathcal{S}}_{-i, d} \rightarrow 2^{\tilde{\mathcal{S}}_{i,d}}$ is the constraint correspondence as defined in \eqref{eq:constraintCorrespondance}.

\noindent
Furthermore, for any feasible set of probability distributions $\mathcal{C}\subseteq \Delta(\mathcal{A})$, we define
\begin{equation}\label{eq:RdTildeInducedByFeasibleDistributions}
    \tilde{\mathcal{R}}_d
    = 
    \{
        \bm{\gamma}\in\tilde{\mathcal{S}}_d 
        \mid 
        \bm{p}_{\bm{\gamma}}\in\mathcal{C}
    \}
\end{equation}
and 
$\tilde{\mathcal{K}}_{i,d}(\bm{\gamma}_{-i})
= 
\{
    \gamma_i\in\tilde{\mathcal{S}}_{i,d} 
    \mid 
    (\gamma_i,\bm{\gamma}_{-i})\in\tilde{\mathcal{R}}_d
\}
$ (as in \eqref{eq:definitionKi}), 
implying, 
for any $\gamma_{-i}\in\tilde{\mathcal{S}}_{-i,d}$
\begin{equation} 
    \tilde{\mathcal{K}}_{i,d}(\bm{\gamma}_{-i}) 
    = 
    \{
        \gamma_i\in\tilde{\mathcal{S}}_{i,d} 
        \mid 
        \bm{p}_{(\gamma_i,\bm{\gamma}_{-i}))}\in\mathcal{C}
    \}
\end{equation}
Remark that in Section \ref{sec:constrainedProbabilities}, the constraint set $\mathcal{R}_d$ depends on the probability distribution of the correlated strategy profile whereas in \eqref{eq:RdTildeInducedByFeasibleDistributions} the constraint set $\tilde{\mathcal{R}}_d$ depends on the probability distribution induced by the correlated strategy profile as defined in \eqref{eq:probaDistributionInduced}.

In the next theorem, we show that if the constraints are generated by a feasible set of probability distributions, an additional independent randomization by the players is not necessary in terms of correlated equilibrium distributions, implying that the extensions of $G$ are sufficient. Before showing this result, we give the following preliminary lemma.
\begin{lemma}\label{lem:2}
    Let $d = (\Omega, (\mathcal{P}_i)_{i\in\mathcal{N}}, \bm{q})$ be a correlation device. For every  correlated strategy profile $\bm{\gamma} \in \tilde{\mathcal{S}}_{d}$, there exists a canonical correlation device $d_c$ and a correlated strategy profile ${\bm{\alpha}} \in {\mathcal{S}}_{d_c}$ such that $\bm{p}_{\bm{\alpha}}$ = $\bm{p}_{{\bm{\gamma}}}$. Furthermore, for every 
    $\alpha_i^\prime : \mathcal{A} \rightarrow \mathcal{A}_i$, there exists
    $\gamma_i^\prime : \Omega \rightarrow \Delta(\mathcal{A}_i)$
    such that
    $\bm{p}_{(\alpha_i^\prime, \bm{\alpha}_{-i})} = \bm{p}_{(\gamma_i^\prime, \bm{\gamma}_{-i})}$.
\end{lemma}
\begin{proof}[\textbf{Proof of Lemma \ref{lem:2}}]
Let $d = (\Omega, (\mathcal{P}_i)_{i\in\mathcal{N}}, \bm{q})$ be a correlation device, 
$\bm{\gamma} \in \tilde{\mathcal{S}}_{d}$ a correlated strategy 
and 
$d_c = (\mathcal{A}, (\mathcal{P}'_i)_{i\in\mathcal{N}}, \bm{q}_c)$ the canonical correlation device
such that for any $\bm{a} \in \mathcal{A}$,
\begin{equation}\label{eq:distribution_canonical_device}
    \bm{q}_c(\bm{a}) = \sum_{\omega \in \Omega} \bm{q}(\omega)
\prod_{j \in \mathcal{N}} {\gamma}_j(\omega)(a_j)
\end{equation}
Furthermore, let $\bm{\alpha} : \mathcal{A} \rightarrow \mathcal{A}$ be the identity correlated strategy, \ie $\bm{\alpha} = \bm{id}$.
Then, 
for any $\bm{a} \in \mathcal{A}$,
\begin{equation}\label{eq:DistributionEquality}
    \bm{p}_{\bm{\alpha}}(\bm{a})
    = \bm{q}_c(\bm{a})     
    = \bm{p}_{\bm{\gamma}}(\bm{a}) 
\end{equation}
Therefore, there exists a canonical correlation device $d_c$ and a correlated strategy $\bm{\alpha}$ such that $\bm{p}_{\bm{\alpha}} = \bm{p}_{\bm{\gamma}}$ which concludes the first part of the proof. 
We now show the second part of the lemma. 
For any ${\alpha}^\prime_i : \mathcal{A} \rightarrow \mathcal{A}_i$, we have the probability distribution of $(\alpha^\prime_i, \bm{\alpha}_{-i}) = (\alpha^\prime_i, \bm{id}_{-i})$ such that for any $\bm{a} \in \mathcal{A}$
    \begin{align}
        \bm{p}_{(\alpha^\prime_i, \bm{\alpha}_{-i})}(\bm{a})
        & =\bm{p}_{(\alpha^\prime_i, \bm{id}_{-i})}(\bm{a})
        \\
        & = 
        \sum_{\bm{b} \in \mathcal{A}} \bm{p}_{\bm{id}}(\bm{b}) \mathds{1}_{\{(\alpha^\prime_i, \bm{id}_{-i})(\bm{b}) = \bm{a}\}} 
        \\
        & = 
        \sum_{{b}_{i} \in \mathcal{A}_i} \bm{p}_{\bm{id}}(b_i, \bm{a}_{-i}) \mathds{1}_{\{\alpha^\prime_i({b}_i, \bm{a}_{-i}) = {a}_i\}} 
        \label{eq:distribution_deviation_from_alpha_i}
    \end{align}
Let $\gamma_i^\prime:\Omega\rightarrow \Delta(\mathcal{A}_i)$ be the strategy such that for any $\omega \in \Omega$ and any $a_i \in \mathcal{A}_i$,
\begin{equation}\label{eq:gamma_i_prime}
    \gamma_i^\prime(\omega)(a_i) = \sum_{b_i \in \mathcal{A}_i} \gamma_i(\omega)(b_i) \mathds{1}_{\{ \alpha_i^\prime(b_i, \bm{a}_{-i}) = a_i \}} 
\end{equation}
The strategy $\gamma_i^\prime$ is well-defined since for any $\omega \in \Omega$ and any $a_i \in \mathcal{A}_i$, $\gamma_i^\prime(\omega)(a_i) \geq 0$ and, for any $\omega \in \Omega$,
\begin{align}
    \sum_{a_i \in \mathcal{A}_i} \gamma_i^\prime(\omega)(a_i)
    & = \sum_{a_i \in \mathcal{A}_i} \sum_{b_i \in \mathcal{A}_i} \gamma_i^\prime(\omega)(b_i) \mathds{1}_{\{ \alpha_i^\prime(b_i, \bm{a}_{-i}) = a_i \}} \\
    & = \sum_{b_i \in \mathcal{A}_i} \gamma_i^\prime(\omega)(b_i) \sum_{a_i \in \mathcal{A}_i} \mathds{1}_{\{ \alpha_i^\prime(b_i, \bm{a}_{-i}) = a_i \}} \\
    & = 1
\end{align}
where the last equality is obtained from $\sum_{a_i \in \mathcal{A}_i} \mathds{1}_{\alpha_i^\prime(b_i, \bm{a}_{-i}) = a_i} = 1$ which holds since for every $b_i \in \mathcal{A}_i$, there exists a unique $a_i \in \mathcal{A}_i$ such that $\alpha_i^\prime(b_i, \bm{a}_{-i}) = a_i$.

\noindent 
Then, we have the probability distribution $\bm{p}_{(\gamma^\prime_i, \bm{\gamma}_{-i})}$ induced by $(\gamma_i^\prime, \bm{\gamma}_{-i})$ such that, for any $\bm{a} \in \mathcal{A}$,
\begin{align}
    \bm{p}_{(\gamma^\prime_i, \bm{\gamma}_{-i})}(\bm{a})
     & = 
    \sum_{\omega \in \Omega} \bm{q}(\omega) \gamma^\prime_i(\omega)(a_i) \left( \prod_{j \neq i} \gamma_j(\omega)(a_j) \right) \\
     & = \sum_{\omega \in \Omega} \bm{q}(\omega) \sum_{b_i \in \mathcal{A}_i} \gamma_i(\omega)(b_i) \mathds{1}_{\{ \alpha_i^\prime(b_i, \bm{a}_{-i}) = a_i \}} \left( \prod_{j \neq i} \gamma_j(\omega)(a_j) \right) \\
     & = \sum_{b_i \in \mathcal{A}_i} \sum_{\omega \in \Omega} \bm{q}(\omega) \gamma_i(\omega)(b_i) \left( \prod_{j \neq i} \gamma_j(\omega)(a_j) \right) \mathds{1}_{\{ \alpha_i^\prime(b_i, \bm{a}_{-i}) = a_i \}} \label{eq:}
\end{align}
where the second line follows from using \eqref{eq:gamma_i_prime}. 

\noindent 
Furthermore, since
\begin{equation}
    \sum_{\omega \in \Omega} \bm{q}(\omega) \gamma_i(\omega)(b_i) \left( \prod_{j \neq i} \gamma_{j}(\omega)(a_j) \right) = \bm{p}_{\bm{\gamma}}(b_i, \bm{a}_{-i})
\end{equation}
Then for any $\bm{a} \in \mathcal{A}$, we have
\begin{equation}         
     \bm{p}_{(\gamma^\prime_i, \bm{\gamma}_{-i})}(\bm{a})
     = \sum_{b_i \in \mathcal{A}_i} \bm{p}_{\bm{\gamma}}(b_i, \bm{a}_{-i}) \mathds{1}_{\{ \alpha_i^\prime(b_i, \bm{a}_{-i}) = a_i \}}
\end{equation} 
Thus, from equations \eqref{eq:DistributionEquality} and  \eqref{eq:distribution_deviation_from_alpha_i}, we have
\begin{equation}
     \bm{p}_{(\gamma^\prime_i, \bm{\gamma}_{-i})}(\bm{a}) 
     = \sum_{b_i \in \mathcal{A}_i} \bm{p}_{\bm{id}}(b_i, \bm{a}_{-i}) \mathds{1}_{\{ \alpha_i^\prime(b_i, \bm{a}_{-i}) = a_i \}} 
     = \bm{p}_{(\alpha_i^\prime, \bm{\alpha}_{-i})}(\bm{a})
\end{equation} 
This result concludes the proof.
\end{proof}


\begin{theorem}\label{thm:mixedStrategies}
    Let $G$ be a finite non-cooperative game, $\mathcal{C}$ a feasible set of probability distributions, $d$ a correlation device and $\bm{\gamma}^*\in\tilde{\mathcal{S}}_d$ a correlated strategy profile.  
    If $(d,\tilde{\mathcal{R}}_d,\bm{\gamma}^*)$ is a constrained correlated equilibrium of $\Delta G$, then it exists a constrained correlated equilibrium $(d^\prime,\mathcal{R}_{d^\prime},\bm{\alpha}^*)$ of $G$ such that $\bm{p}_{\bm{\gamma}^*} = \bm{p}_{\bm{\alpha}^*}$.
\end{theorem}
\begin{proof}[\textbf{Proof of Theorem \ref{thm:mixedStrategies}}]
    Let $\mathcal{C}$ be a feasible set of probability distributions and $(d,\mathcal{R}_d,\bm{\gamma}^*)$ be a constrained correlated equilibrium of $\Delta G$. For any $i \in \mathcal{N}$ 
    and
    any $\gamma_i^\prime \in \tilde{\mathcal{S}}_{i, d}$
    such that
    $\bm{p}_{(\gamma_i^\prime,\bm{\gamma}^*_{-i})} \in \mathcal{C}$, 
    we have, 
    \begin{equation}\label{eq:equilibriumCondition}
        \sum_{\bm{a} \in \mathcal{A}} \bm{p}_{\bm{\gamma}^*}(\bm{a}) u_i(\bm{a}) \geq \sum_{\bm{a} \in \mathcal{A}} \bm{p}_{(\gamma_i^\prime,\bm{\gamma}^*_{-i})}(\bm{a}) u_i(\bm{a})
    \end{equation}
    From Lemma \ref{lem:2} and its proof, there exists  a canonical correlation device $d_c$  and a correlated strategy profile $\bm{\alpha}^* \in \mathcal{S}_{d_c}$ such that $\bm{\alpha}^*= \bm{id}$ and $\bm{p}_{\bm{\alpha}^*} = \bm{p}_{\bm{\gamma}^*}$. 
    Then, 
    for any $\gamma_i^\prime \in \tilde{\mathcal{S}}_{i, d}$ such that $\bm{p}_{(\gamma_i^\prime,\bm{\gamma}^*_{-i})} \in \mathcal{C}$,
    \begin{equation}
        \sum_{\bm{a} \in \mathcal{A}}
        \bm{p}_{\bm{\alpha}^*}(\bm{a}) u_i(\bm{a}) \geq \sum_{\bm{a} \in \mathcal{A}} \bm{p}_{(\gamma_i^\prime,\bm{\gamma}^*_{-i})}(\bm{a}) u_i(\bm{a}) \label{eq:ineq}
    \end{equation}
    Applying again Lemma \ref{lem:2}, 
    for any $\alpha^\prime_i \in {\mathcal{S}}_{i, d_c}$ such that $\bm{p}_{(\alpha_i^\prime,\bm{\alpha}^*_{-i})} \in \mathcal{C}$, 
    it exists $\bm{\gamma}_i^\prime\in\tilde{\mathcal{S}}_{i,d}$ such that 
    $\bm{p}_{(\gamma_i^\prime,\bm{\gamma}_{-i}^*)} = \bm{p}_{(\alpha_i^\prime,\bm{\alpha}_{-i}^*)} \in\mathcal{C}$, implying,
    \begin{equation}
        \sum_{\bm{a} \in \mathcal{A}} \bm{p}_{(\alpha_i^\prime,\bm{\alpha}^*_{-i})}(\bm{a}) u_i(\bm{a})
        = 
        \sum_{\bm{a} \in \mathcal{A}} \bm{p}_{(\gamma_i^\prime,\bm{\gamma}^*_{-i})}(\bm{a}) u_i(\bm{a})
    \end{equation}
    Thus, using inequality \eqref{eq:ineq}, we obtain that
    for any $\alpha^\prime_i \in {\mathcal{S}}_{i, d_c}$ such that $\bm{p}_{(\alpha_i^\prime,\bm{\alpha}^*_{-i})} \in \mathcal{C}$, 
    \begin{equation}
        \sum_{\bm{a} \in \mathcal{A}}
        \bm{p}_{\bm{\alpha}^*}(\bm{a}) u_i(\bm{a}) 
        \geq 
        \sum_{\bm{a} \in \mathcal{A}} \bm{p}_{(\alpha_i^\prime,\bm{\alpha}^*_{-i})}(\bm{a}) u_i(\bm{a}).
    \end{equation}
It follows that $(d_c,\mathcal{R}_{d_c},\bm{\alpha^*})$ is a constrained correlated equilibrium since  $\bm{p}_{\bm{\alpha}^*}= \bm{p}_{\bm{\gamma}^*}\in \mathcal{C}$.
 \end{proof}

The latter theorem shows that the set of constrained correlated equilibrium distributions of $\Delta G$ is a subset of the set of constrained correlated equilibrium distributions of $G$. 
The following example shows that this inclusion can be strict.
\medskip

\noindent \textbf{Example.} 
Consider the game in Figure \ref{tab:tab2} and let $\mathcal{C} = \{ \bm{q} \in \Delta(\mathcal{A}) \mid \bm{q}(A,C) >0 \}$ be the feasible set of probability distributions and $d = (\Omega, (\mathcal{P}_i)_{i\in\mathcal{N}}, \bm{q})$ a correlation device such that $\Omega = \{\omega\}$.
\begin{figure}[H]
    \normalsize 
    \centering
    \scalebox{1.2}{
        \begin{tabular}{c|c|c}
                & \multicolumn{1}{c|}{$C$}  & $D$  \\ \hline
            $\,\,\,A\,\,\,$ & \multicolumn{1}{c|}{$(0, 0)$} & $(1, 2)$  \\ \hline
            $\,\,\,B\,\,\,$ & \multicolumn{1}{c|}{$(2, 1)$} & $(3, 0)$
        \end{tabular}}
    \caption{Two-player game.}
    \label{tab:tab2}
\end{figure}
\noindent Let $\bm{\alpha}^* = (\alpha^*_1, \alpha^*_2)$ be a correlated strategy profile such that $\alpha_1^*(\omega) = A$ and $\alpha_2^*(\omega)= C$. We have $\bm{\alpha}^* \in \mathcal{R}_{d} = \{ \bm{\alpha} \in \mathcal{S}_d \mid \bm{p}_{\bm{\alpha}} \in \mathcal{C} \}$ and the triplet $(d, \mathcal{R}_d, \bm{\alpha}^*)$ is a constrained correlated equilibrium of $G$. We show that no triplet $(d^\prime, \tilde{\mathcal{R}}_{d^\prime}, \bm{\gamma}^\prime)$ where $\tilde{\mathcal{R}}_{d^\prime} = \{ {\bm{\gamma}} \in \tilde{\mathcal{S}}_d \mid \bm{p}_{{\bm{\gamma}}} \in \mathcal{C} \}$ such that $\bm{p}_{\bm{\gamma}^\prime} = \bm{p}_{\bm{\alpha}^*}$, is not a constrained correlated equilibrium  of $\Delta G$.\\
Let $d^\prime = (\Omega^\prime, (\mathcal{P}^\prime_i)_{i\in\mathcal{N}}, \bm{q}^\prime)$ be a correlation device and $\bm{\gamma}^\prime \in\tilde{\mathcal{S}}_{d^\prime}$ a correlated strategy profile such that $\bm{p}_{\bm{\gamma}^\prime} = \bm{p}_{\bm{\alpha}^*}$. Thus, $\bm{\gamma}^\prime$ satisfies for every $\omega^\prime \in \Omega^\prime$,
\begin{equation}
    \gamma_1^\prime(\omega^\prime) = (1 \cdot A, 0 \cdot B) \text{ and } \gamma_2^\prime(\omega^\prime) = ({1} \cdot C, 0 \cdot D).
\end{equation}
Players' utilities using $\bm{\gamma}^\prime$ are $(u_1(\bm{\gamma}^\prime),u_2(\bm{\gamma}^\prime)) = (0,0)$.
Consider now an alternate strategy for player 1 $\bar{\gamma}_1 \in \tilde{\mathcal{S}}_{1,d}$ such that for any $\omega^\prime \in \Omega^\prime$,
\begin{equation}
    \bar{\gamma}_1(\omega^\prime) = (\sfrac{1}{2} \cdot A, \sfrac{1}{2} \cdot B)
\end{equation} 
Then the utility of player 1 will be $u_1(\bar{\gamma}_1,\gamma_2^\prime) = 1 > 0$.  Notice  that    the probability distribution $\bm{p}_{(\bar{\gamma}_1, \gamma^\prime_2)}$ induced by $(\bar{\gamma}_1,\gamma_2^\prime)$ is in $\mathcal{C}$, implying that $(\bar{\gamma}_1,\gamma_2^\prime)\in\tilde{\mathcal{R}}_{d^\prime}$ Then the triplet $(d^\prime, \tilde{\mathcal{R}}_{d^\prime}, \bm{\gamma}^\prime)$ is not a constrained correlated equilibrium since $\bar{\gamma}_1$  is a profitable feasible deviation for player $1$. 

In this example, although the probability distribution $\bm{p}_{\bm{\gamma}^\prime}$ is a constrained correlated equilibrium distribution of $G$, it is not a constrained correlated equilibrium distribution of $\Delta G$.

\section{Numerical experiments}\label{sec:illustration}
In this section, we provide numerical results supporting the relevance of the concept of constrained correlated equilibria and show some of its properties in the case of constraints on probability distributions studied in Section~\ref{sec:constrainedProbabilities}. 

Let $G$ be the game of Chicken shown in Figure~\ref{fig:chickenGame}. 
The set of correlated equilibrium distributions of $G$ and corresponding utilities are shown in Figure~\ref{fig:polytope}.
Let $\mathcal{C}$ be the feasible set of probability distributions guaranteeing some level of social welfare \cite{nash1950bargaining} such that,
\begin{align}\label{eq:FeasibleSet}
    \mathcal{C} = \{ \bm{q} \in \Delta(\mathcal{A}) \mid \sum_{i\in\mathcal{N}} u_i(\bm{q}) \geq SW_{min} \}
\end{align}
where $SW_{min}\in\mathbb{R}$ is interpreted as a minimum level of social welfare.
We show in Figure~\ref{fig:SW_12} (a) the set of correlated equilibrium distributions (the polytope of correlated equilibrium distributions in yellow is not visible due to the superposition of the sets but it can be seen in Figure~\ref{fig:polytope} and Figure~\ref{fig:SW_14}), the set of feasible distributions (green) and the set of constrained correlated equilibrium distributions (red) for the constraints $\mathcal{C}$ with $SW_{min} = 12$. Figure~\ref{fig:SW_12} (b) displays the corresponding set of pairs of utilities.

\begin{figure}
    \begin{subfigure}[t]{0.5\textwidth}
        \centering
        \hspace*{-0.5cm}\includegraphics[scale=0.11]{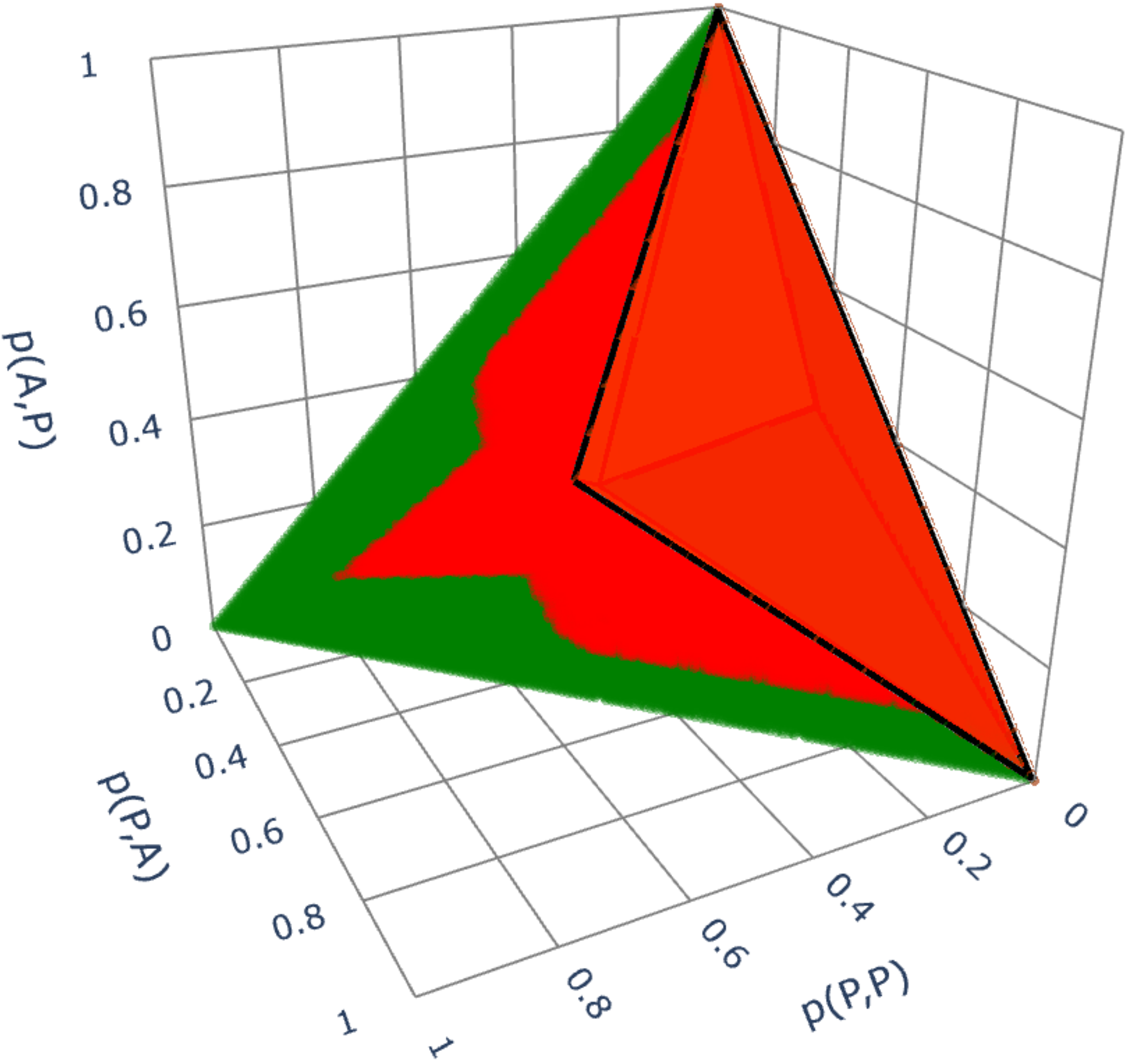}
        \caption{Probability distributions}
  \end{subfigure}
  \begin{subfigure}[t]{0.5\textwidth}
      \centering
      \hspace*{0.1cm}\includegraphics[scale=0.26]{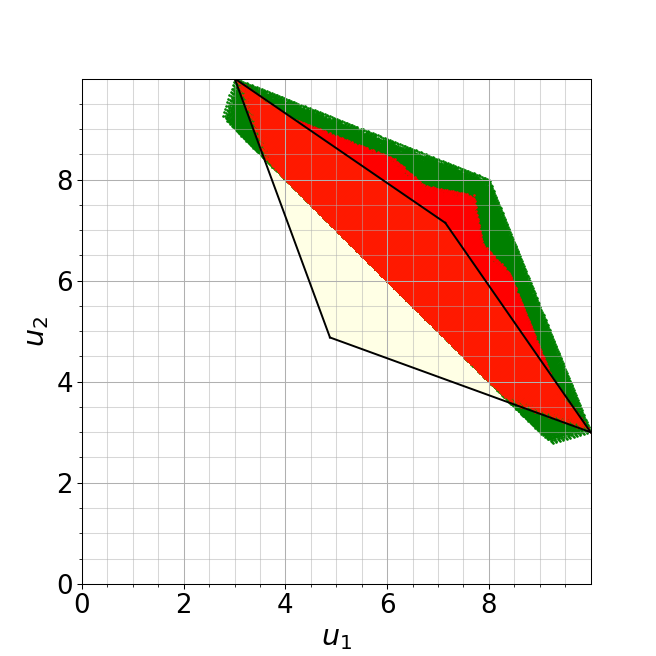}
      \caption{Utilities} 
  \end{subfigure}
  \caption{(a) Sets of correlated equilibria $\mathcal{D}$ (yellow), feasible distributions $\mathcal{C}$ with $SW_{min} = 12$ (green) and constrained correlated equilibrium distributions (red). (b) Utilities induced by correlated equilibria (yellow), feasible distributions (green) and constrained correlated equilibria (red).}
  \label{fig:SW_12}
\end{figure}

First, Figure~\ref{fig:SW_12} (a) shows that unlike the set of correlated equilibrium distributions, the set of constrained correlated equilibrium distributions is not necessarily convex.
Second, there are correlated equilibrium distributions in the polytope and the feasible set. These distributions are constrained equilibrium distributions as shown by Proposition~\ref{prop:4}.
Hence, the set of constrained correlated equilibria contains the intersection of the polytope and the feasible region.
Third, there are constrained correlated equilibrium distributions outside the set of correlated equilibrium distributions, showing that constraints (in this case) stabilize some probability distributions with at least one player having a profitable deviation.
Taking the constraint into account, it appears that for such constrained correlated equilibrium distributions, although a unilateral deviation is profitable for a player, the resulting probability distribution decreases the social welfare below the threshold $SW_{min}$ implying that 
such deviation is not feasible.
Indeed, in spite of increasing the utility of the player, this deviation induces a distribution decreasing the utility of other players such that the social welfare is below $SW_{min}$.

Figure~\ref{fig:SW_12} (b) shows that some constrained correlated equilibrium distributions induce a payoff $(7.70, 7.70)$ which is strictly greater than the maximum achievable without constraints $(7.14, 7.14)$ for each player.


\begin{figure}
    \begin{subfigure}[t]{0.5\textwidth}
   \centering
    \hspace*{-0.5cm}\includegraphics[scale=0.11]{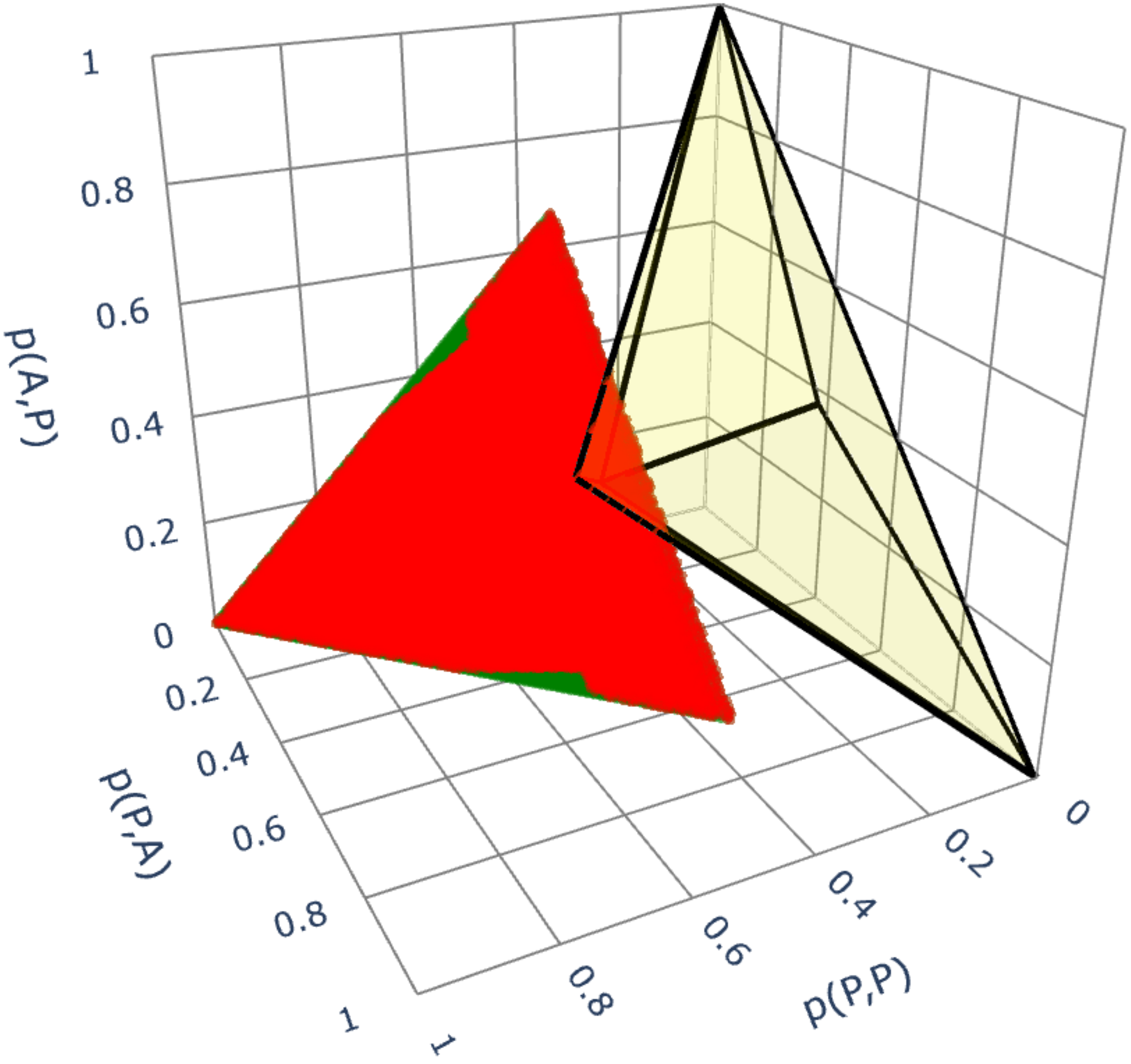}
    \caption{Probability distributions}
  \end{subfigure}
  \begin{subfigure}[t]{0.5\textwidth}
  \centering
    \hspace*{0.1cm}\includegraphics[scale=0.26]{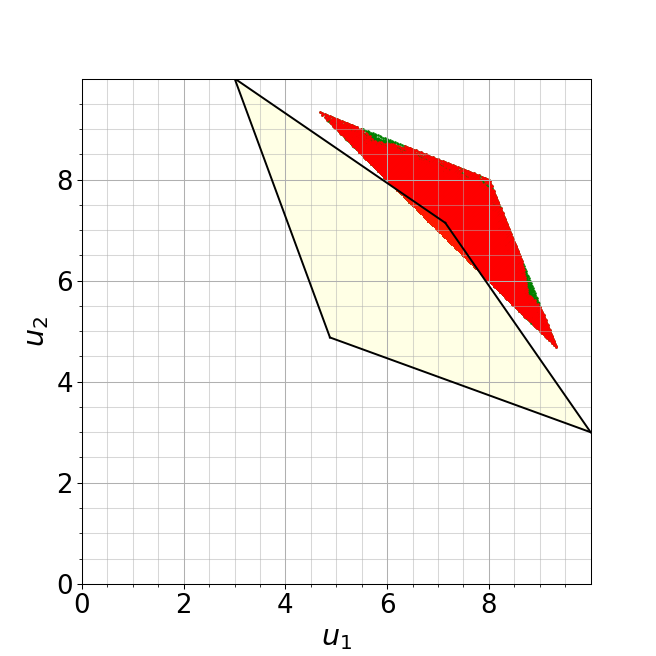}
    \caption{Set of induced payoffs} 
  \end{subfigure}
  \caption{(a) Set of correlated equilibria $\mathcal{D}$ (yellow), feasible distributions $\mathcal{C}$ with $SW_{min} = 14$ (green) and constrained correlated equilibria (red). (b) Set of payoffs induced by correlated equilibria (yellow), feasible distributions (green) and constrained correlated equilibria (red).}
  \label{fig:SW_14}
\end{figure}
Figure~\ref{fig:SW_14} shows the probability distributions and corresponding utilities for a higher minimum social welfare $SW_{min}=14$ inducing a feasible set of probability distributions included in the feasible set of probability distributions obtained with $SW_{min}=12$. According to Corollary~\ref{cor:SmallerConstraints},
this inclusion implies that any probability distribution which is a constrained correlated equilibrium distribution for $SW_{min} = 12$ and feasible for $SW_{min} = 14$ is also a constrained correlated equilibrium distribution for $SW_{min} = 14$. Furthermore, we note that the action profile $(P,P)$ is a constrained correlated equilibrium which was not the case in the previous example. 
This constrained correlated equilibrium yielding a utility $(8, 8)$ is also a generalized Nash equilibrium.
\begin{figure}
     \begin{subfigure}[t]{0.5\textwidth}
   \centering
    \hspace*{-0.5cm}\includegraphics[scale=0.11]{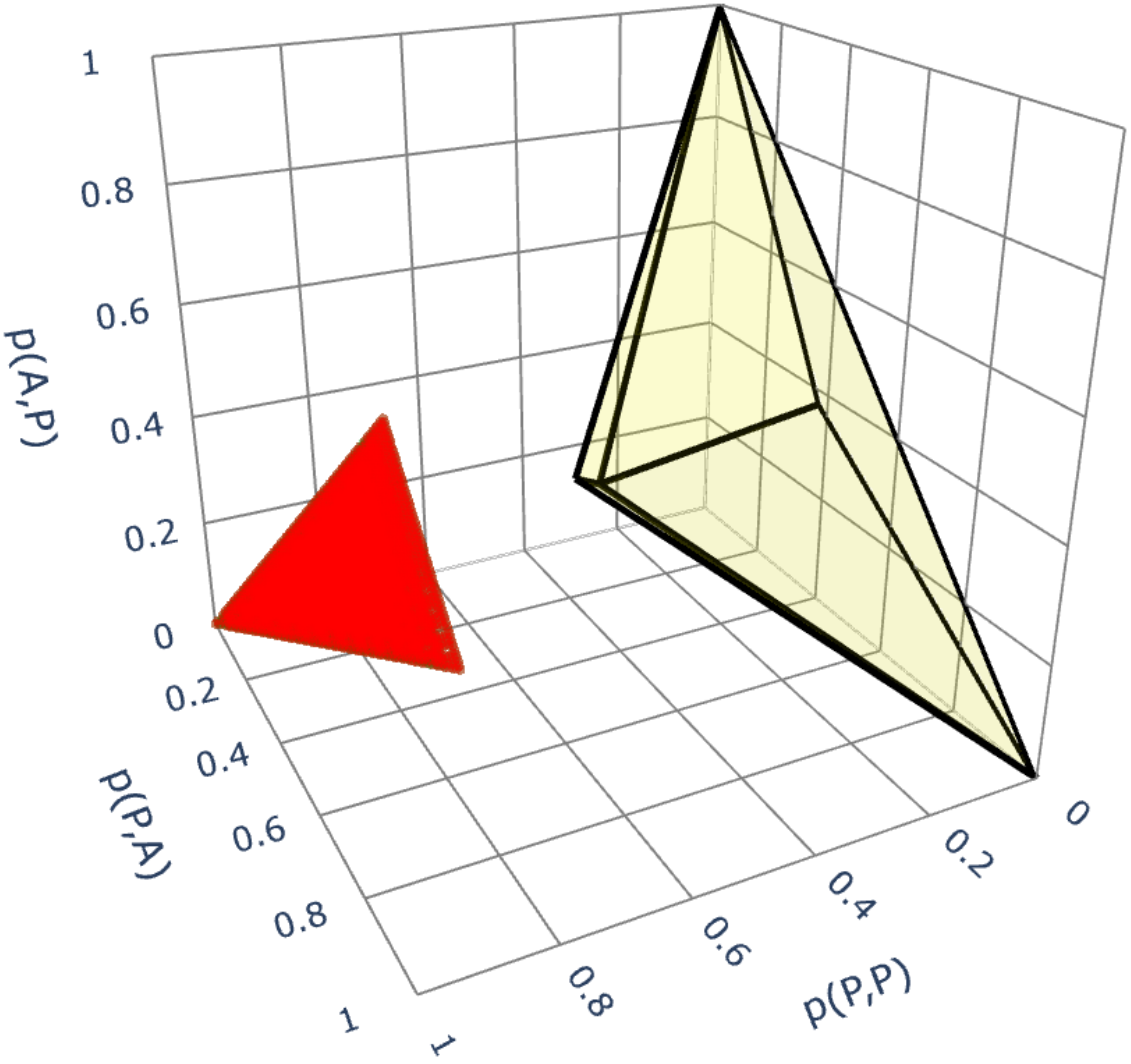}
    \caption{Probability distributions}
  \end{subfigure}
  \begin{subfigure}[t]{0.5\textwidth}
  \centering
    \hspace*{0.1cm}\includegraphics[scale=0.26]{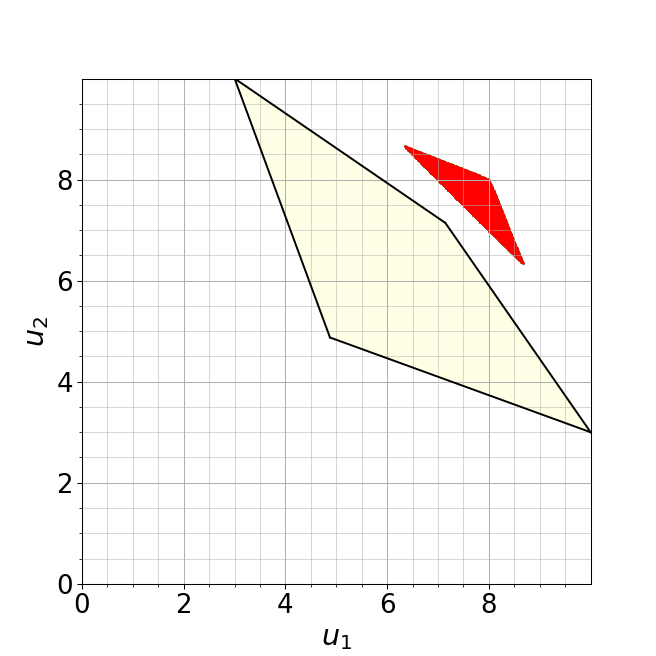}
    \caption{Utilities} 
  \end{subfigure}
  \caption{(a) Set of correlated equilibria $\mathcal{D}$ (yellow), feasible distributions $\mathcal{C}$ with $SW_{min} = 15$ (green) and constrained correlated equilibria (red). (b) Set of payoffs induced by correlated equilibria (yellow), feasible distributions (green) and constrained correlated equilibria (red).}
  \label{fig:SW_15}
\end{figure}

Finally, setting the threshold to $SW_{min} = 15$ gives a feasible set of probability distributions such that the intersection with the polytope of correlated equilibrium distributions is empty. 
In this case, the set of constrained correlated equilibrium probability distributions only consists of feasible distributions such that at least one player can improve her utility by playing an infeasible deviation as shown in Proposition~\ref{prop:alternativeDefinition}.

\section{Conclusion}\label{sec:conclusions}
This paper defines the concept of constrained correlated equilibrium for finite non-cooperative games generalizing correlated equilibria. In the general case of arbitrary constraints, sufficient equilibrium conditions are proved.  Among other results, it is shown that any correlated equilibrium satisfying the constraints,  is a constrained correlated equilibrium.
In the particular case of constraints generated by a feasible set of probability distributions over action profiles, we have shown that canonical devices are sufficient to characterize the set of constrained correlated equilibrium distributions (as in the case without constraints). 
In terms of existence, it is shown that convexity and compactness of the feasible set of probability distributions imply the existence of constrained correlated equilibria. 
Furthermore, it is shown using examples that the constrained correlated equilibrium distributions may not belong to the polytope of correlated equilibrium distributions.
Finally, we have shown that the set of constrained correlated equilibrium distributions of the mixed extension is a subset of the set of equilibrium distributions of the game in pure strategies.\\ 

Future directions of research include extensions to infinite games and more general correlation devices as well as a detailed analysis of the connection to Bayesian rationality. 
In fact, as a solution concept, correlated equilibrium can be justified as "an Expression of Bayesian Rationality" \cite{aumann1987correlated} but
it is still unclear if this result also holds for constrained correlated equilibria. 
Furthermore, the existence problem should be studied for weaker or alternative assumptions.

From the learning perspective, the approachability of the set of constrained correlated equilibria should be studied as well as the relevance of the concept for learning with constraints, particularly with respect to learning dynamics as correlated equilibria for regret-based learning \cite{hart_adaptive_2005}.  Indeed, the learning developed in the literature is based on the concept of regret that uses  only part of the   deviation thanks to the equivalent definition of correlated equilibrium given by  (\ref{eq:expostconditon}). 
 This is no longer valid in the presence of coupled constraints, which calls for a study to exploit existing algorithms to approach the set of constrained correlated equilibria.  


\bibliographystyle{abbrv}
\bibliography{bibliography}



\section*{Appendix}


\begin{lemma}\label{lem:continuity}
   Let $G = (\mathcal{N}, (\mathcal{A}_i)_{i \in \mathcal{N}}, (u_i)_{i \in \mathcal{N}})$ be a finite non-cooperative game. The function $\bm{z}_{\beta_i, \bm{p}}$ is continuous in $\bm{p}$. 
\end{lemma}


\begin{proof}[\textbf{Proof of Lemma \ref{lem:continuity}}]
To show that the function $\bm{z}_{\beta_i, \bm{p}}$ is continuous in $\bm{p}$, it is sufficient to show that it is $K$-Lipschitz where $K=|\mathcal{A}_i| \sqrt{|\mathcal{A}|} $ with $|\mathcal{B}|$ denoting the cardinality of a set $\mathcal{B}$.  
We have for all $\bm{p}$ and $\bm{p}^\prime$ in $\Delta(\mathcal{A})$
\begin{align}\label{eq:maxUtility}
        ||\bm{p} - \bm{p}^\prime||
        =
        \sqrt{\sum_{\bm{a} \in \mathcal{A}} [\bm{p}(\bm{a}) - \bm{p}^\prime(\bm{a})]^2}
        \geq
         \max_{\bm{a} \in \mathcal{A}}|\bm{p}(\bm{a})-\bm{p}^\prime(\bm{a})|
    \end{align}
    where  $||.||$ is the Euclidean norm. We have
 \begin{align}
        ||\bm{z}_{\beta_i,\bm{p}} - \bm{z}_{\beta_i,\bm{p}^\prime}|| 
                            & = \sqrt{\sum_{\bm{a} \in \mathcal{A}}[\bm{z}_{\beta_i,\bm{p}}(\bm{a}) - \bm{z}_{\beta_i,\bm{p}^\prime}(\bm{a})]^2}\\
                            & = \sqrt{\sum_{\bm{a} \in \mathcal{A}}\left[\sum_{b_i \in \mathcal{A}_i} \left(\bm{p}(b_i, \bm{a}_{-i}) - \bm{p}^\prime(b_i, \bm{a}_{-i})\right)\mathds{1}_{\beta_i(b_i) = a_i}\right]^2}\\
                            &\leq
                            \sqrt{\sum_{\bm{a} \in \mathcal{A}}\left[\sum_{b_i \in \mathcal{A}_i} \max_{\bm{a}^\prime} |\bm{p}(\bm{a}^\prime)-\bm{p}^\prime(\bm{a}^\prime)|\right]^2}\\
                            &\leq
                            \sqrt{\sum_{\bm{a} \in \mathcal{A}}\left[\sum_{b_i \in \mathcal{A}_i} ||\bm{p} - \bm{p}^\prime||\right]^2}\label{norm}\\
                            &= | \mathcal{A}_i| \sqrt{|\mathcal{A}|} \times ||\bm{p} - \bm{p}^\prime||,
    \end{align}
   where the inequality in (\ref{norm}) holds from (\ref{eq:maxUtility}). This concludes the proof of the Lemma.
\end{proof}

\ifCLASSOPTIONcaptionsoff
  \newpage
\fi

\end{document}